\documentclass[10pt,journal,comsoc]{IEEEtran}
\usepackage{amsmath,amsfonts}
\usepackage{algorithm}
\usepackage{array}
\usepackage[caption=false,font=normalsize,labelfont=sf,textfont=sf]{subfig}
\usepackage{textcomp}
\usepackage{stfloats}
\usepackage{url}
\usepackage{verbatim}
\usepackage{graphicx}
\usepackage{cite}
\usepackage{etex}
\usepackage{dcolumn}
\usepackage{epsfig}
\usepackage{amssymb}
\usepackage{amsthm}
\usepackage{hyperref}
\usepackage{setspace}
\usepackage[svgnames]{xcolor} 
\usepackage{pstricks,pst-node,pst-plot,pstricks-add}
\usepackage{authblk}
\usepackage{algpseudocode}
\usepackage[noabbrev]{cleveref}

\usepackage{tikz}

\usepackage{mathdots}
\usepackage{mathrsfs}
\usepackage{cancel}
\usepackage{color}
\usepackage{siunitx}
\usepackage{array}
\usepackage{multirow}
\usepackage{gensymb}
\usepackage{tabularx}
\usepackage{booktabs}
\usepackage{xspace}

\usepackage[referable]{threeparttablex}

\input{Definition.def}
\graphicspath{{figs/}}
\newcommand{\figref}[1]{Fig.\,\ref{#1}}
\newcommand{\tabref}[1]{Table~\ref{#1}}
\tikzset{every picture/.style={line width=0.75pt}}
\def\BibTeX{{\rm B\kern-.05em{\sc i\kern-.025em b}\kern-.08em
		T\kern-.1667em\lower.7ex\hbox{E}\kern-.125emX}}
 \newcommand{\light}{\texttt{LightSecAgg}\xspace}
\newcommand{\google}{\texttt{SecAgg}\xspace}
\newcommand{\googlep}{\texttt{SecAgg+}\xspace}
\newcommand{\turbo}{\texttt{TurboAgg}\xspace}
\newcommand{\fast}{\texttt{FastSecAgg}\xspace}
\newcommand{\swift}{\texttt{SwiftAgg+}\xspace}
\newcommand{\ccesa}{\texttt{CCESA}\xspace}

\begin{document}

\title{\swift: Achieving Asymptotically Optimal Communication Loads in Secure Aggregation for Federated Learning}

\author{Tayyebeh Jahani-Nezhad, Mohammad Ali Maddah-Ali,~\IEEEmembership{Senior Member,~IEEE}, Songze Li, Giuseppe Caire,~\IEEEmembership{Fellow,~IEEE,}
\thanks{T.~Jahani-Nezhad and G.~Caire are with the Electrical Engineering and Computer Science Department, Technische Universit\"at Berlin, 10587 Berlin, Germany (e-mail: t.jahani.nezhad@tu-berlin.de, caire@tu-berlin.de) }       
\thanks{T.~Jahani-Nezhad and M.~A.~Maddah-ali are with the Department of Electrical Engineering, Sharif University of Technology, Tehran 11365-11155, Iran (e-mail: maddah\_ali@sharif.edu)}
\thanks{S.~Li is with IoT Thrust, The Hong Kong University of Science and Technology (Guangzhou), and
Department of Computer Science and Engineering, The Hong Kong University of Science and Technology (e-mail:songzeli@ust.hk)}
\thanks{The work of T. Jahani-Nezhad and G. Caire was partially funded by the European Research Council under the ERC Advanced Grant N. 789190, CARENET.}}



\maketitle

\begin{abstract}
We propose \swift,  a novel secure aggregation protocol for federated learning systems, where a central server aggregates local models of $N \in  \mathbb{N}$ distributed users, each of size $L \in \mathbb{N}$, trained on their local data, in a privacy-preserving manner. \swift can significantly reduce the communication overheads without any compromise on security, and achieve optimal communication loads within diminishing gaps.  Specifically, in presence of at most $D=o(N)$ dropout users,  \swift achieves a per-user communication load of $(1+\mathcal{O}(\frac{1}{N}))L$ symbols and a server communication load of $(1+\mathcal{O}(\frac{1}{N}))L$ symbols,  with a worst-case information-theoretic security guarantee, against any subset of up to $T=o(N)$ semi-honest users who may also collude with the curious server. Moreover, the proposed \swift allows for a flexible trade-off between communication loads and the number of active communication links.
In particular, for $T<N-D$ and for any $K\in\mathbb{N}$, \swift can achieve the server communication load of $(1+\frac{T}{K})L$ symbols, and per-user communication load of up to $(1+\frac{T+D}{K})L$ symbols, where the number of pair-wise active connections in the network is $\frac{N}{2}(K+T+D+1)$.  
\end{abstract}


\section{Introduction}
\IEEEPARstart{F}{ederated} learning (FL) is an emerging distributed learning framework that allows a group of distributed users (e.g., mobile devices) to collaboratively train a global model of size $L$ with their local private data, without sharing the data~\cite{mcmahan2017communication,kairouz2019advances,li2020federated}. 
	Specifically, in a FL system with a central server and $N$ users, during each training iteration, the server sends the current state of the global model to the users. Receiving the global model, each user then trains a local model with its local data, and sends the local model to the server. By aggregating the local models, the server can update the global model for the next iteration. 
	While the local datasets are not directly shared with the server, several studies have shown that a curious server can launch model inversion attacks to reveal information about the training data of individual users from their local models (see, e.g.,~\cite{zhu2020deep,geiping2020inverting}). Therefore, the key challenge to protect users' data privacy is to design \emph{secure aggregation} protocols, which allow the aggregation of local models to be computed without revealing each individual model even if the server colludes with a limited number of users, called semi-honest users, in the system. Moreover, as some users may randomly drop out of the aggregation process (due to low batteries or unstable connections), the server should be able to robustly recover the aggregated local models of the surviving users in a privacy-preserving manner. 
	
As such motivated, a secure aggregation protocol \google is proposed in~\cite{bonawitz2017practical}. In \google, each pair of users agree on a shared random vector. In addition, each user generates a private random vector and secret shares it with other users. 
Before sending its local model to the server, each user masks the local model in a particular way. The mask is created based on the private and shared random vectors such that the shared part can be canceled out when aggregated at the server. In order to recover the result and cancel the remaining masks, the private random vector of each surviving user and the shared random vectors of each dropped user must be reconstructed using the collected shares from surviving users in the second round of communication. 

Note that to have an information-theoretic security, the shared and the private random vectors must be chosen uniformly at random from vectors of size $L$. However, to make it practical in \google, it is assumed that the users agree on common seeds for a pseudo-random generator rather than on the entire shared random vectors which is not secure in the context of information theory. In this paper, we focus on information-theoretic security guarantees.





One of the major challenges for \google to scale up is the communication load. The secret sharing among users requires all-to-all communication, which incurs quadratic cost in the number of users $N$. There has been a series of works that aim to improve the communication efficiency of \google~(see, e.g.,~\cite{so2021turbo,bell2020secure,choi2020communication}). In~\cite{so2021turbo}, \turbo is proposed to perform secure aggregation following a circular topology, achieving a communication load of $\mathcal{O}(LN\log N)$ at the server and $\mathcal{O}(L\log N)$ at each user. The idea of \turbo is partitioning the users into several groups each of size $\mathcal{O}(\log N)$ to accelerate model aggregation phase. To protect individual models, each user masks its local model in a particular way.
 Furthermore, the masked local models are encoded using Lagrange coded computing \cite{yu2019lagrange} such that the data of dropped users can be recovered at each stage of the protocol. 
 Each user in a group computes the aggregation of the masked models and the aggregation of coded models of all users in the previous groups, and sends them to all users in the next group, as well as the masked and coded version of its model.  The missing terms of aggregated masked models, caused by the dropped users, can be recovered in the next group if at least half of the users in a group do not drop out. \turbo guarantees the privacy of individual local models  with high probability as long as the number of semi-honest users is less than $N/2$.

 In \cite{bell2020secure}, \googlep is proposed which considers a $k$-regular communication graph among users instead of the complete graph used in \google to reduce the communication and computation loads, where $k=\mathcal{O}(\log N)$. It means that each user secret shares with a subset of other users instead of sharing  with all of them. It is shown that \googlep requires a communication of $\mathcal{O}(LN+ N\log N)$ at the server and $\mathcal{O}(L+ \log N)$ at each user. In~\cite{choi2020communication}, another similar idea, \ccesa, is proposed in which a sparse random graph is used as communication network instead of the complete graph. Thus, secret sharing is used for a subset of users instead of all-to-all communication. While \googlep, \turbo and \ccesa improve the communication efficiency of \google, they only provide probabilistic privacy guarantees as opposed to the worst-case guarantee of \google.

Some other secure aggregation protocols 
focus on reducing the computation complexity of \google, rather than the communication load. \fast is proposed in \cite{kadhe2020fastsecagg} which presents a multi-secret sharing scheme based on a finite field version of the Fast Fourier Transform to achieve a trade-off between the number of secrets, the dropout tolerance and the privacy threshold. \fast reduces the computation load of the server as well as per-user to $\mathcal{O}(L\log N)$, while achieves the same order of communication as \google with lower privacy threshold and dropout tolerance.

In ~\cite{yang2021lightsecagg}, \light is proposed to reduce the computation complexity bottleneck in secure aggregation. In \light, each user independently samples a random mask, encodes it using $T$-private MDS (maximum distance separable) coding, generating secret shares of the mask to send to other users. Then, each user adds mask to its local model and sends the masked model to the server. The server can reconstruct and cancel out the aggregated mask of surviving users via a one-shot decoding, using the aggregated coded masks received in a second phase of communication with surviving users. Similar idea of one-shot reconstruction of aggregated mask was also utilized in \cite{zhao2021information}, which nevertheless requires a trusted third party, and much more randomness needs to be generated, as well as more storage at each user. 

As an earlier and conference version of this manuscript, in \cite{jahani2022swiftagg}
, we propose \texttt{SwiftAgg} which reduces the server communication load via partitioning users into several groups and using Shamir's secret sharing \cite{shamir1979share} to keep the local models private. \texttt{SwiftAgg} is able to achieve correct aggregation in presence of up to $D$ dropout users, with formally-proven the worst-case security guarantee against up to $T$ semi-honest users.  Independently, similar ideas have been used in \texttt{CodedSecAgg} \cite{schlegel2021codedpaddedfl}, but dropout users in groups have not been addressed and there is no formal definition and proof of privacy. The communication loads and the computation loads of the existing approaches are compared in Table~\ref{table}.

In this paper, we propose \swift as a new scheme for secure aggregation in federated learning, which
is robust against up to $D$ user dropouts and achieves worst-case security guarantees against up to $T<N-D$ semi-honest users which may collude with each other as well as with the server to gain some information. In \swift, in order to compute the aggregation of local models in a privacy-preserving manner,
 the number of symbols received by the server is $(1+\mathcal{O}(\frac{1}{N}))L$ as long as $T=o(N)$ and $N-D=\mathcal{O}(N)$. In addition, the number of symbols sent by each user is $(1+\mathcal{O}(\frac{1}{N}))L$, for $T=o(N)$ and $D=o(N)$. 

Moreover, since in federated learning many of the users may not be accessible from each other, the number of active communication links, among users or from users to the server, is also an important factor when designing secure aggregation protocols. \swift allows for a flexible trade-off between the communication loads and the network connections. To achieve this trade-off, we first partition the users into groups of size $T+D+K$ for some $K\in\mathbb{N}$ (See Fig.~\ref{fig3}), and the groups are arranged on an arbitrary hierarchical tree with the server as the root. Then in the first phase, each user partitions its local model into $K$ parts. In the second phase, users within each group secret share the partitions of their local models using ramp secret sharing \cite{blakley1984security}, and aggregate the shares locally. In the third phase, each user aggregates  in-group shares and the messages received from the corresponding users in its children groups and sends the result to the corresponding user in its parent group, which is finally sent to the server. 
In the third phase, if one user in a group drops out, the rest of the users in the sequence that includes the dropped user remain silent. This aggregation tree provides flexibility to adjust the delay of each iteration of training process based on the network connections. Additionally, parameter $K$ controls the trade-off between communication loads and network connectivity which can be chosen based on the constraints of the network.


\swift simultaneously achieves the following advantages compared to the existing works:
	\begin{enumerate}
        \item
         For $T=o(N)$ and $N-D=\mathcal{O}(N)$, the server communication load in \texttt{SwiftAgg+} is within factor 1 of the cut-set lower bound.
         \item \swift has the flexibility to reduce the number of active connections between the users at the cost of increasing the communication load away from the optimum.
         
		\item
    	It is resilient to up to $D$ user dropouts. As opposed to the existing schemes, in \swift, due to the structure of the scheme, the effects of user dropouts are already handled in one round of communication with the server, and there is no need to take place another round of communication.
		\item
		It achieves worst-case information-theoretic security against a curious server and any subset of up to $T < N-D$ semi-honest users.
		\item
		Based on the possible connections among users, \swift have a flexibility to control the delay of each training iteration.
	\end{enumerate}
	
\tabref{table} compares  different schemes for  secure aggregation problem in terms of the communication and computation loads. For a fair comparison, we consider two metrics: the server communication load and the per-user communication load. Server communication load indicates the aggregated size of all messages which are sent or received by the server, and per-user communication load denotes the aggregated size of all messages which are sent by each user. Furthermore, server and per-user computations indicate the computation loads of the server and each user respectively. In this table, the  computational complexity of the schemes in \cite{yang2021lightsecagg,so2021turbo,jahani2022swiftagg}, and \swift,  is calculated based on Reed-Solomon decoding complexity \cite{kedlaya2011fast}. Any improvements in this complexity can be applied to those schemes.

Compared with existing schemes, as shown in \tabref{table}, \swift significantly reduces the server communication load and per-user communication load, and for $T=o(N)$ and $N-D=\mathcal{O}(N)$,
it requires $L(1+\mathcal{O}(\frac{1}{N}))$ symbols for the server communication load and for $T=o(N)$ and $D=o(N)$, it requires $L(1+\mathcal{O}(\frac{1}{N}))$ symbols for the per-user communication load.
Meanwhile, similar as \light, the proposed \swift also enjoys small server computation complexity with worst-cast security guarantees.

	\begin{table*}[t]
		\centering
		\caption{ Communication and computation loads of secure aggregation frameworks in federated learning. Here $N$ is the total number of the users, $L$ is the model size, $T$ is the number of semi-honest users, $D$ is the number of dropouts, $s$ is the length of seeds for pseudo random generator.In this table, the last column indicates the maximum number of semi-honest users and dropouts that the scheme can tolerate.}
		\label{table}
		\resizebox{2\columnwidth}{!}{
		\begin{threeparttable}
			\begin{tabular}{||c |c c c c| c||} 
				\hline
				Approach & Server communication  & Per-user communication  & Server computation & Per-user computation & Threshold\\ [0.5ex] 
				\hline\hline
				\google~\cite{bonawitz2016practical} & $\mathcal{O}(NL+sN^2)$ & $\mathcal{O}(L+sN)$ & $\mathcal{O}(LN^2)$ & $\mathcal{O}(LN+sN^2)$ & $T,D\le \ceil{N/3}-1$ \\
				\hline
				\googlep~\cite{bell2020secure}& $\mathcal{O}(NL+sN\log N)$ & $\mathcal{O}(L+s\log N)$ & $\mathcal{O}(LN\log N + N\log^2N)$ & $\mathcal{O}(L\log N + s\log^2N)$ & $\frac{TN}{N-1}<N-D$ \\
				\hline
				
				\turbo~\cite{so2021turbo} & $\mathcal{O}(NL\log N)$ & $\mathcal{O}(L\log N)$ & $\mathcal{O}(L\log N\log^2{\log N})$ & $\mathcal{O}(L\log N\log^2{\log N})$ & $T,D< N/2$\\ 
				\hline
				\ccesa~\cite{choi2020communication}& $\mathcal{O}(NL+sN\sqrt{N\log N})$ & $\mathcal{O}(L+s\sqrt{N\log N})$ & $\mathcal{O}(LN\log N)$& $\mathcal{O}(L\sqrt{N\log N}+sN\log N)$& \text{Server is only curious}\\
				\hline
				\fast~\cite{kadhe2020fastsecagg} & $\mathcal{O}(NL+N^2)$  & $\mathcal{O}(L+N)$  & $\mathcal{O}(L\log N)$  & $\mathcal{O}(L\log N)$ & $T,D<N/2$\\
				\hline
				\rule{0pt}{12pt}
				\light~\cite{yang2021lightsecagg}& $\mathcal{O}(NL)$ & $(2+\frac{T+D+1}{N-T-D})L$& $\mathcal{O}\big(\frac{N-D-1}{N-T-D}L\log^2 (N-D-1)\big)$ & $\mathcal{O}\big(L(2+\frac{T}{N-T-D}+\frac{N}{N-T-D}\log^2 (N-D-1))\big)$& $T<N-D$\\
				\hline 
				\rule{0pt}{12pt}
		        \texttt{SwiftAgg}\cite{jahani2022swiftagg} & $(1+T){L}$ & $(T+D+1){L}$& $\mathcal{O}\big(TL\log^2(T)\big)$ &  $\mathcal{O}\big( {L}(T+D+1)(1+\log^2(T))\big)$& $T<N-D$\\ 
			
				\hline 
				\rule{0pt}{12pt}
				Proposed \swift & $(1+\frac{T}{N-T-D}){L}$ & $(1+\frac{T+D}{N-T-D}){L}$& $\mathcal{O}\big(\frac{N-D-1}{N-T-D}L\log^2(N-D-1)\big)$ &  $\mathcal{O}\big( {L}(\frac{N}{N-T-D}+\frac{N}{N-T-D}\log^2(N-D-1))\big)$& $T<N-D$\\ 
				\hline
			\end{tabular}
			 \begin{tablenotes}
      \item In this table, in order to have a fair comparison in terms of communication load, the design parameter in \light is chosen such that it minimizes the communication loads.
    \end{tablenotes}
			\end{threeparttable}
			}
	\end{table*}
	
	\noindent {\bf Notation:}  For $n\in\mathbb{N}$ the notation $[n]$ represents set $\{1,\dots,n\}$. The notation $a|b$ means that $a$ divides $b$.
	Furthermore, the cardinality of set $\mathcal{S}$ is denoted by $|\mathcal{S}|$. In addition, we denote the difference of two sets $\mathcal{A}$, $\mathcal{B}$ as $\mathcal{A}\backslash\mathcal{B}$, that means the set of elements which belong to $\mathcal{A}$ but not $\mathcal{B}$. $H(X)$ denotes the entropy of random variable $X$ and $I(X;Y)$ is the mutual information of two random variables $X$ and $Y$.

\section{Problem formulation}
	We consider the secure aggregation problem, for a federated learning system, consisting of a server and $N$ users ${U}_1,\dots,{U}_N$. For each $n \in [N]$, user $n$ has a  private local model of size $L$, denoted by  $\mathbf{W}_n\in\mathbb{F}^{L}$, for some finite field $\mathbb{F}$. Consider that the local models are from some joint distribution 
     $\mathbf{W}_1,\mathbf{W}_2,\dots,\mathbf{W}_N\sim P_{\mathbf{W}_1,\mathbf{W}_2,\dots,\mathbf{W}_N}(\mathbf{W}_1,\mathbf{W}_2,\dots,\mathbf{W}_N)$ which is often unknown.
Each entry of the local models, after a universal linear (affine) mapping, is represented with a non-negative integer number,  less than $\ell$ for some $\ell \in \mathbb{N}$.
 Each user $n$ also has a collection of random variables $\mathcal{Z}_n$, whose elements are selected uniformly at random from $\mathbb{F}^{L}$, and independently from each other and from the local models. 
	Users can send messages to each other and also to the server, using error-free private communication links. 
	$\mathbf{M}^{(L)}_{n\to n'} \in \mathbb{F}^* \cup \{ \perp \}$ denotes  the message that user $n$ sends to user $n'$.  In addition, $\mathbf{X}^{(L)}_n \in \mathbb{F}^* \cup \{ \perp \}$ denotes the message sent by node $n$ to the server.  The null symbol $\perp$ represents the case no message is sent.

	The message $\mathbf{M}^{(L)}_{n\to n'}$ is a function 
	of $\mathbf{W}_n$, $\mathcal{Z}_n$, and the messages that node $n$ has received from other nodes so far. We denote the corresponding encoding function by $f^{(L)}_{n\to n'}$. 
	Similarly, $\mathbf{X}^{(L)}_n$ is a function of $\mathbf{W}_n$, $\mathcal{Z}_n$,  and the messages that node $n$ has received from other nodes so far. We denote the corresponding encoding function by $g^{(L)}_{n}$. For a subset ${\cal S} \subseteq [N]$, we let  $\mathcal{X}_{\mathcal{S}}=\{\mathbf{X}_n^{(L)}\}_{n\in\mathcal{S}}$ represent the set of messages the server receives from users in $\mathcal{S}$.
	We assume that a subset $\mathcal{D}\subset [N]$ of users drop out, i.e.,
	stay silent (or send $\perp$ to other nodes and the server) during the protocol execution. We denote the number of dropped out users as $D=|\mathcal{D}|$.
	
	We also assume that a subset
	$\mathcal{T}\subset[N]$ of the users, whose identities are not known, are semi-honest. It means that users in $\mathcal{T}$ follow the protocol faithfully; however, they are curious and may collude with each other or with the server to gain information about the local models of the honest users. We assume $|\mathcal{T}|\le T$, for some known security parameter $T< N-D$. 
	
	A secure aggregation scheme consists of the encoding functions  $f^{(L)}_{n\to n'}$ and $g^{(L)}_{n}$, $n,n' \in [N]$, such that the following conditions are satisfied: 
	
	
	\textbf{1. Correctness:} The server is  able to recover $\mathbf{W}=\sum_{n\in[N]\backslash\mathcal{D}}{\mathbf{W}}_{n}$,  using 
	$\mathcal{X}_{[N]\backslash\mathcal{D}}~=~\{\mathbf{X}^{(L)}_n\}_{n \in [N]\backslash\mathcal{D}}$. More
	precisely,
	\begin{align}
	H\bigg(\sum\limits_{n\in[N]\backslash\mathcal{D}}\mathbf{W}_n\big| \mathcal{X}_{[N]\backslash\mathcal{D}}\bigg)=0.
	\end{align}
	

	\textbf{2. Privacy Constraint:} 
	For any joint distribution 
 $P_{\mathbf{W}_1,\mathbf{W}_2,\dots,\mathbf{W}_N}(\mathbf{W}_1,\mathbf{W}_2,\dots,\mathbf{W}_N)$, after receiving $\mathcal{X}_{[N]\backslash\mathcal{D}}$ and colluding with semi-honest users in $\mathcal{T}$, the server should not gain any information about local models of the honest users, beyond the aggregation of them, and beyond what it infers from their correlation with the local models of the semi-honest users. Formally, 
	\begin{align}\label{serve-prv}
	\nonumber
	I\bigg(&\mathbf{W}_n, {n\in[N]\backslash\mathcal{T}}; \mathcal{X}_{[N]\backslash\mathcal{D}},\bigcup\limits_{k\in\mathcal{T}}\{\mathbf{M}^{(L)}_{k'\to k}, k'\in[N]\}
	\bigg|\\& \sum\limits_{n\in[N]\backslash\{\mathcal{D}\cup\mathcal{T}\}}{\mathbf{W}_n},\{\mathbf{W}_k,\mathcal{Z}_k,{k\in\mathcal{T}}\}\bigg)=0.
	\end{align}
	
	For a secure aggregation scheme satisfying the above two conditions, we define per-user communication load and server communication load as follows: 
	
	\begin{definition}[Normalized average per-user communication load] denoted by $R^{(L)}_{\text{user}}$,   is defined as the aggregated size of all messages sent by users, normalized by $NL$, i.e.,
		\begin{align*}
		 	R^{(L)}_{\text{user}} =\frac{1}{NL} \sum_{\substack{n\in [N],\\n'\in[N]\backslash n}}\big( H(\mathbf{M}^{(L)}_{n\to n'})+H(\mathbf{X}^{(L)}_{n})\big).  
		\end{align*}
	\end{definition}
	\begin{definition}[Normalized server communication load] denoted by $R^{(L)}_{\text{server}}$,  is defined as the the aggregated size of all messages received by the server, normalized by $L$, i.e.,
	\begin{align*}
	    R^{(L)}_{\text{server}}=\frac{1}{L} \sum_{n\in [N]} H(\mathbf{X}^{(L)}_{n}).
	\end{align*}
	\end{definition}
	We say that the pair of $(R_{\text{server}}, R_{\text{user}})$
	is achievable, if there exist a sequence of secure aggregation schemes with rate tuples $(R^{(L)}_{\text{server}}, R^{(L)}_{\text{user}})$, $L=1,2,\ldots$, such that 
	\begin{align*}
	    R_{\text{server}}=\limsup_{L \rightarrow \infty} R^{(L)}_{\text{server}},\hspace{3mm}
	    R_{\text{user}}=\limsup_{L \rightarrow \infty}R^{(L)}_{\text{user}}.
	\end{align*}
	\begin{definition}[Communication graph for an achievable scheme] For an achievable scheme $\mathscr{A}(N,T,D)$ for secure aggregation problem, consisting of $N$ users and a server, we define the undirected communication graph $\mathscr{G}_{\mathscr{A}}(\mathcal{V},\mathcal{E}_{\mathscr{A}})$. In this graph, $\mathcal{V} =\{ U_1, \ldots, U_N, \textsf{server} \}$ is the set of vertices representing users and the server,  and $\mathcal{E}_{\mathscr{A}}$ is the set of the edges corresponding to active communication links among these vertices. In this graph, edge $e\in\mathcal{E}_{\mathscr{A}}$ between two vertices $v_1,v_2\in\mathcal{V}$ exists if the corresponding users communicate directly with each other (they do not send null message $\perp$). 
	\end{definition}

\section{Main results}
	
	In this section, we present our main results on per-user and server communication loads, achieved by the proposed \swift scheme. Note that  the operations are done in a finite field,  which is large enough to avoid hitting the boundary in the process of aggregation.  We choose a finite field $GF(p)$ denoted by $\mathbb{F}_p$, for some prime number $p$,  where $N(\ell-1) < p \leq 2N (\ell-1)$. 


\begin{theorem} \label{theorem}
		Consider a secure aggregation problem, with $N$ users and one server,  where up to $T$ users are semi-honest and up to $D$ users may drop out. There is an achievable scheme which need
	  \begin{align}
	  \label{upper}
	 &R_{\text{server}}=  \left(1+\frac{T}{N-T-D}\right), \\\nonumber
	  &\ R_{\text{user}}\le \left(1+\frac{T+D}{N-T-D}\right),
	  \end{align}
    symbols from $\mathbb{F}_p$.
	\end{theorem}
	
	\begin{proof}
		The proof can be found in Section \ref{proof_comm}.
	\end{proof}
	
		To achieve the communication loads of \eqref{upper} in Theorem \ref{theorem}, we propose \swift, a novel secure aggregation scheme in federated learning.
		In \swift, under the condition of $T=o(N)$ and $D=o(N)$, the number of symbols sent by each user is $(1+\mathcal{O}(\frac{1}{N}))L$ and the number of symbols received by the server is $(1+\mathcal{O}(\frac{1}{N}))L$ in order to compute the aggregation of local models in a privacy-preserving manner. In other words, for large $N$, the communication load per user tends to $L$ symbols from $\mathbb{F}_p$ and the total communication load of the server also tends to $L$ symbols from $\mathbb{F}_p$.
 \begin{theorem}\label{th-server}
     For $T=o(N)$ and $N-D=\mathcal{O}(N)$, in terms of the number of required bits, the server communication load in \texttt{SwiftAgg+} is within factor 1 of the cut-set lower bound.
 \end{theorem}
 \begin{proof}
     The proof can be found in Section \ref{proof-th-server}.
 \end{proof}
 \begin{theorem}\label{th-user}
    For local models with uniform distributions, in terms of the number of required bits, the per-user communication load in  \texttt{SwiftAgg+} is within factor $\log_{\ell}{\ell N} $ of the cut-set lower bound, for $T=o(N)$ and $D=o(N)$.
 \end{theorem}
  \begin{proof}
     The proof can be found in Section \ref{proof-th-user}.
 \end{proof}
 \begin{remark}
      Note that  the case $T = o(N)$  is well justified. Restricted with a limited budget, it often becomes prohibitively difficult for the adversary to hack a constant fraction of the users as the number of users grows. For example, it is seems easier to hack one of only two users rather than 50 out of 100 ones.
 \end{remark}
	As the number of active communication links is important in federated learning, \swift also allows us to reduce the number of active network connections in the secure aggregation problem at the cost of increasing the communication load away from the optimum. The result is presented in the following theorem.
	
	\begin{theorem}
	   Consider a secure aggregation problem, with $N$ users and one server,  where up to $T$ users are semi-honest and up to $D$ users may drop out. For any $K\in \mathbb{N}$, where $(T+D+K)|N$, there exist an achievable scheme ${\mathscr{A}_K}(N,T,D)$ with
		\begin{align}\label{comm_flix}
		R_{\text{server}}^{(\mathscr{A}_K)}&=1+\frac{T}{K},\\\nonumber
		 R_{\text{user}}^{(\mathscr{A}_K)}&\le \left(1+\frac{T+D}{K}\right),
		\end{align}
		symbols from $\mathbb{F}_p$. The communication graph of this achievable scheme $\mathscr{G}_{\mathscr{A}_K}(\mathcal{V},\mathcal{E}_{\mathscr{A}_K})$ consists of $|\mathcal{E}_{\mathscr{A}_K}|=\frac{N}{2}(K+T+D+1)$ edges.
	\end{theorem}
	\begin{proof}
		The proof can be found in Section \ref{proof_comm}.
	\end{proof}
    \begin{remark}
	In \eqref{comm_flix} by increasing parameter $K$, both per-user communication load and  server communication load are reduced, while the number of active connections in the achievable scheme $|\mathcal{E}_{\mathscr{A}_K}|$ is increased and vice versa. This allows us to have a flexible trade-off between communication loads and the complexity of the communication pattern (in terms of number of active links). 
	\end{remark}

    To achieve the trade-off between communication loads and network connections $|\mathcal{E}_{\mathscr{A}_K}|$, \swift partitions the users into some disjoint groups, each of size $K+D+T$ users labeled as user $1$ to user $K+D+T$ of that group. The $\frac{N}{K+D+T}$ groups are arranged in an arbitrary tree with the server as the root. The model aggregation starts from the leaves of the tree, and proceeds in three main phases: Phase 1: Each user partitions the vector of its local model of length $L$ into $K$ sub-vectors of length $\frac{L}{K}$, Phase 2:
	Each user uses a ramp sharing~\cite{blakley1984security} to share the sub-vectors of its model to other users within its group,  Phase 3: Each user aggregates the shares that it receives in Phase 2 and also adds it to the messages received from the corresponding users in its children groups and sends the result to the corresponding user in its parent group. These messages are finally reach to the server and are used to recover the aggregated local models. 
	

\begin{remark}
The minimum communication loads in Theorem \ref{theorem} are achieved by setting $K$ to its maximum value of $N-D-T$. The minimum number of communication connections is achieved by setting $K=1$, which reduces to the \texttt{SwiftAgg} scheme in \cite{jahani2022swiftagg}.
\end{remark}

\begin{remark}
While in problem formulation we let $L\to\infty$, the proposed scheme works for finite values of $L$, where $K|L$. If $K\nmid L$, we can zero-pad the local models.
\end{remark}

\section{Illustrative Examples}

In this section, we present two illustrative examples of the proposed \swift scheme to solve a secure aggregation problem, under different choices of the design parameter $K$. This also demonstrates the trade-off between communication loads and the number of communication links achieved by varying $K$ in \swift.

We consider a secure aggregation problem over a federated learning system of one server and $N=12$ users, $U_1, U_2, \dots, U_{12}$. At most $D=1$ user may dropout during the aggregation process, and up to $T=2$ semi-honest users may collude with each other to gain some information about the local models of other users. Each user $n$ has a local model $\mathbf{W}_n$ which is a vector of size $L$,	

\subsection{Miminum communication loads with $K=9$}


Each user $n$ samples uniformly at random two vectors $\mathcal{Z}_n=\{\mathbf{Z}_{n,1},\mathbf{Z}_{n,2} \}$ from $\mathbb{F}^{\frac{L}{9}}$, and then takes the following steps:
	\begin{enumerate}
	\item {\bf Partitioning the Local Models: }
		User $n$ partitions its local model into $K=9$ parts, i.e.,
		\begin{align*}
		    \mathbf{W}_n=[\mathbf{W}_{n,1},\mathbf{W}_{n,2},\dots,\mathbf{W}_{n,9}]^T,
		\end{align*}
		where each part $\mathbf{W}_{n,k}, k\in[9]$ is a vector of size $\frac{L}{9}$.
		\item {\bf Secret Sharing and Aggregation:}
		User $n$ forms the following polynomial.
		\begin{align*}
		\mathbf{F}_{n}(x)=\mathbf{W}_{n,1}+\mathbf{W}_{n,2}x+\dots+\mathbf{W}_{n,9}x^8+\mathbf{Z}_{n,1}x^9+\mathbf{Z}_{n,2}x^{10},
		\end{align*}
		where the coefficients of the first $K=9$ terms are the partitions of the local model of user $n$, $n\in[12]$. 
		
		 Let $\alpha_t\in \mathbb{F}$, $t\in [12]$, be distinct non-zero constants. We assign   $\alpha_t$ to user $t$.  In this step, each user $n$ sends the evaluation of its polynomial function at $\alpha_{t}$, i.e.,  $\mathbf{F}_{n}(\alpha_{t})$, to user $t$, for $t\in [12]$. \figref{complete_graph} represents the connection links among the users and server. As shown in \figref{complete_graph}, $\mathscr{G}(\mathcal{V},\mathcal{E})$ for this example is a complete graph which means all users communicate with each other and to the server.
		If a user $m$ drops out and stays silent, $\mathbf{F}_{m}(\alpha_{t})$ is just presumed to be zero.
		
		Each user $n$ calculates 
		$\mathbf{S}_{n}=\sum_{n'=1}^{12}\mathbf{F}_{n'}(\alpha_{n})$. In this example, assume that $U_3$  drops out and does not send its share to other users. Other users treat its share as zero. In this phase, at most 132 communications take place, each of size $\frac{L}{9}$. 
		
		
		\item 	{\bf Communication with the Server:}
		User $n$ sends $\mathbf{S}_{n}$ to the server, for $n\in [12]$.  
		Clearly in this example, user 3 remains silent and sends nothing (or null message $\perp$) to the server. 
		
		\item {\bf Recovering the result:}
		Let us define 
		\begin{align*}
		\mathbf{F}(x)\hspace{-1mm}\triangleq\hspace{-1mm}\sum_{\substack{n=1 \\ n\neq 3}}^{12}\mathbf{F}_{n}(x)\hspace{-1mm}=\hspace{-1mm}
		\sum\limits_{k=1}^9x^{k-1}\hspace{-1mm}\sum_{\substack{n=1 \\ n\neq 3}}^{12}\mathbf{W}_{n,k}\hspace{-1mm}+\hspace{-1mm}x^9\hspace{-1mm}\sum_{\substack{n=1 \\ n\neq 3}}^{12}\mathbf{Z}_{n,1}\hspace{-1mm}+\hspace{-1mm}x^{10}\hspace{-1mm}\sum_{\substack{n=1 \\ n\neq 3}}^{12}\mathbf{Z}_{n,2}.
		\end{align*}
		One can verify that  $\mathbf{S}_{n}$, for $n\in[12]\backslash\{3\}$ that are received by the server are indeed equal to $\mathbf{F}(\alpha_1), \mathbf{F}(\alpha_2), \mathbf{F}(\alpha_4)$, \dots, $\mathbf{F}(\alpha_{12})$.
		
		Since $\mathbf{F}(x)$ is a polynomial function of degree 10, based on Lagrange interpolation rule the server can recover all the coefficients of this polynomial using the outcomes of users which are not dropped. In particular, the server can recover $\sum_{\substack{n=1 \\ n\neq 3}}^{12}\mathbf{W}_{n}$ using the coefficients of $x^k, k\in[0:8]$ in the recovered polynomial function. Thus, the server is able to recover the aggregation of local models of surviving users and the correctness constraint is satisfied.
	\end{enumerate} 
	In this example, the per-user communication load is $\frac{4}{3}L$ and the server communication load is $\frac{11}{9}L$. In addition, the communication graph is a complete graph with 78 edges and 13 vertices, where because of user dropout no communication occurs on 12 edges.
	
	\subsection{Trading off communication loads for less connections}
	
	Now we consider a scenario in which not all user can communicate with each other. \swift allows for sacrificing the communication load to reduce the number of required communication links.
	In this case, we take $K=3$, and each user locally samples two vectors $\mathcal{Z}_n=\{\tilde{\mathbf{Z}}_{n,1},\tilde{\mathbf{Z}}_{n,2} \}$ uniformly at random from $\mathbb{F}^{\frac{L}{3}}$.
	To reduce the number of connections, each user takes the following steps:
	\begin{enumerate}
	\item {\bf Partitioning the Local Models:}
		User $n$ partitions its local model into $K=3$ parts, i.e.,
		\begin{align*}
		    \mathbf{W}_n=[\mathbf{W}_{n,1},\mathbf{W}_{n,2},\mathbf{W}_{n,3}]^T,
		\end{align*}
		where each part $\mathbf{W}_{n,k}, k\in[3]$ is a vector of size $\frac{L}{3}$.
		
		\item {\bf Grouping:}
		The set of users are arbitrarily partitioned into $\Gamma=2$ groups of size $\nu\triangleq K+D+T=6$, denoted by $\mathcal{G}_1, \mathcal{G}_2$.  Figure \ref{example2}  represents one example of this partitioning, where $\mathcal{G}_1=\{U_1, U_2, U_3, U_4,U_5,U_6\}$, and $\mathcal{G}_2=\{U_7, U_8, U_9, U_{10},U_{11}, U_{12}\}$. We also order the users in each group arbitrarily. For simplicity of exposition, we may refer to user $n$ based on its location in a group of users. If user $n$ is the $t$th user in group $\gamma$, we call it as user $(\gamma, t)$.  For example in \figref{example2}, user 9 is the same as user $(2,3)$.
		We use indices $n$ or $(\gamma, t)$ interchangeably. 
		
		\item {\bf Intra-group Secret Sharing and Aggregation:}
		User $n$ forms the following polynomial.
		\begin{align*}
		\mathbf{F}_{n}(x)=\mathbf{W}_{n,1}+\mathbf{W}_{n,2}x+\mathbf{W}_{n,3}x^2+\tilde{\mathbf{Z}}_{n,1}x^3+\tilde{\mathbf{Z}}_{n,2}x^4.
		\end{align*}

		Let $\alpha_t\in \mathbb{F}$, $t\in [6]$, be six distinct non-zero constants. We assign   $\alpha_t$ to user $t$ of all groups, i.e., users $(\gamma,t)$, $\gamma=1,2$. 
		
		In this step, each user $(\gamma,t)$  sends the evaluation of its polynomial function at $\alpha_{t'}$, i.e.,  $\mathbf{F}_{(\gamma,t)}(\alpha_{t'})$, to user $(\gamma,t')$, for $t'\in [6]$. For example, in \figref{example2}, user $(2,1 )$, which is indeed $U_7$, sends
		$\mathbf{F}_{7}(\alpha_{1})$, 
		$\mathbf{F}_{7}(\alpha_{2})$, 
		$\mathbf{F}_{7}(\alpha_{3})$,
		$\mathbf{F}_{7}(\alpha_{4})$, 
		$\mathbf{F}_{7}(\alpha_{5})$,
		$\mathbf{F}_{7}(\alpha_{6})$,
		to user $(2,1)$ (or user $U_7$ which is basically itself), user $(2,2)$ (or user $U_8$), user $(2,3)$ (user $U_9$), user $(2,4)$ (user $U_{10}$), user $(2,5)$ (user $U_{11}$), and user $(2,6)$ (user $U_{12}$), respectively. 
		If a user $(\gamma,t)$ drops out and stays silent, $\mathbf{F}_{(\gamma,t)}(\alpha_{t'})$ is just presumed to be zero.
		
		The $t$th user, $t\in[6]$, in $\mathcal{G}_1$  calculates 
		\begin{align*}
		\mathbf{Q}_{(1,t)}=&\mathbf{F}_{(1, 1)}(\alpha_{t})+\mathbf{F}_{(1, 2)}(\alpha_{t})
		+\mathbf{F}_{(1, 3)}(\alpha_{t})+\mathbf{F}_{(1, 4)}(\alpha_{t})\\
		&+\mathbf{F}_{(1, 5)}(\alpha_{t})+\mathbf{F}_{(1, 6)}(\alpha_{t}),
		\end{align*}
		and the $t$th user in $\mathcal{G}_2$ calculates
			\begin{align*}
		\mathbf{Q}_{(2,t)}=&\mathbf{F}_{(2, 1)}(\alpha_{t})+\mathbf{F}_{(2, 2)}(\alpha_{t})
		+\mathbf{F}_{(2, 3)}(\alpha_{t})+\mathbf{F}_{(2, 4)}(\alpha_{t})\\
		&+\mathbf{F}_{(2, 5)}(\alpha_{t})+\mathbf{F}_{(2, 6)}(\alpha_{t}).
		\end{align*}
		In this example, assume that $U_3$ or user $(1,3)$ drops out and does not send its share to other users in the first group. Other users within the group treat its share as zero. In this phase, within each group, at most 30 communication take place, each of size $\frac{L}{3}$. 
		\item  {\bf Inter-group Communication and Aggregation:} In this phase, user $(1,t)$, $t\in [6]$, calculates message
		$\mathbf{S}_{(1,t)}~=~\mathbf{Q}_{(1,t)}$,  and it to user $(2,t)$. 
		
		User $(2,t)$, $t\in[6]$, calculates message $\mathbf{S}_{(2,t)}=\mathbf{S}_{(1,t)}+\mathbf{Q}_{(2,t)}$, upon receiving $\mathbf{S}_{(1,t)}$. If user $(2,t)$ does not receive $\mathbf{S}_{(1,t)}$, it also remains silent for the rest of the protocol. In this particular example that user 3 drops out, it sends no message to user 9, and thus user 9 also remains silent.
		\item 	{\bf Communication with the Server:}
		User $t$ of the last group, i.e., user $(2, t)$ sends $\mathbf{S}_{(2,t)}$ to the server, for $t\in [6]$.  
		Clearly in this example, user 9 remains silent and sends nothing (or null message $\perp$) to the server. 
		
		\item {\bf Recovering the result:}
		Let us define 
		\begin{align*}
		\mathbf{F}(x)\triangleq&\sum_{\substack{n=1 \\ n\neq 3}}^{12}\mathbf{F}_{n}(x)=\sum_{\substack{n=1 \\ n\neq 3}}^{12}\mathbf{W}_{n,1}+x\sum_{\substack{n=1 \\ n\neq 3}}^{12}\mathbf{W}_{n,2}+x^2\sum_{\substack{n=1 \\ n\neq 3}}^{12}\mathbf{W}_{n,3}\\
		&+x^3\sum_{\substack{n=1 \\ n\neq 3}}^{12}\tilde{\mathbf{Z}}_{n,1}+x^4\sum_{\substack{n=1 \\ n\neq 3}}^{12}\tilde{\mathbf{Z}}_{n,2}.
		\end{align*}
		One can verify that  $\mathbf{S}_{(2,t)}$, for $t=1,2,4,5,6$ that are received by the server are indeed equal to $\mathbf{F}(\alpha_t)$, for $t=1,2,4,5,6$.
		
		Since $\mathbf{F}(x)$ is a polynomial function of degree 4, based on Lagrange interpolation rule the server can recover all the coefficients of this polynomial using the outcomes of 5 users which are not dropped. Thus, the server is able to recover the aggregation of local models of surviving users, i.e., $\sum_{\substack{n=1 \\ n\neq 3}}^{12}\mathbf{W}_{n}$ and the correctness constraint is satisfied.
	\end{enumerate}
	In this example, the per-user communication load is $2L$ and the server communication load is $\frac{5}{3}L$. In addition, the communication graph has 42 edges with 13 vertices, where because of user dropouts no communication occurs on 7 edges.
 As compared to Example 1, the communication load is increased while the the number of connections among users are decreased. The number of connections can be minimized by choosing $\nu=T+D+1=4$ and having three groups, where no partitioning of local models is performed. 
 
	
	\begin{figure}
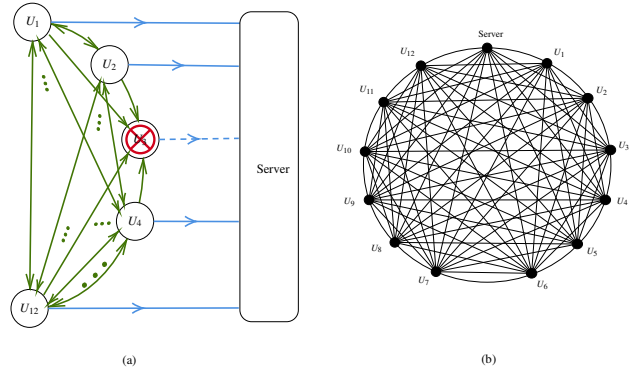


		\centering

		\tikzset{every picture/.style={line width=0.75pt}} 
		\scalebox{0.4}{

\tikzset{every picture/.style={line width=0.75pt}} 


			}
		\caption{(a) Example~1: Minimum communication load and maximum connections for a system consisting of users $U_1,U_2,\dots,U_{12}$, where $T=2$ users are semi-honest and $D=1$ user may drop out. In this example, \swift consists of two main phases: (1) Secret sharing and aggregation, shown by the green directed lines; (2) Communication with the server, shown by the blue directed lines. Dashed lines indicate no communication occurs in this direction. In this example, each user partitions its local model into $N-T-D=9$ parts. (b) Communication graph $\mathscr{G}(\mathcal{V},\mathcal{E})$ of this example.}
		\label{complete_graph}
	\end{figure}
		\begin{figure}
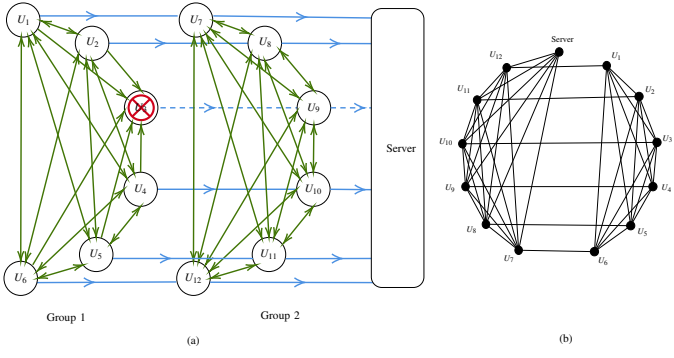


		\centering

		\tikzset{every picture/.style={line width=0.75pt}} 
		\scalebox{0.36}{

\tikzset{every picture/.style={line width=0.75pt}} 


			}
		\caption{(a) Example~2: Trading off communication loads for less connections. Each user partitions its local model into $K=3$ parts. The users are partitioned into two groups of size  $K+T+D=6$. In this example, \swift consists of three main phases: (1) Intra-group Secret sharing, shown by the green directed lines; (2) Inter-group Communication, and (3) Communication with the server shown by the blue directed lines. Dashed lines indicate no communication occurs in this direction. (b) Communication graph $\mathscr{G}(\mathcal{V},\mathcal{E})$ of this example.}
		\label{example2}
	\end{figure}
	\section{The Proposed \swift Scheme}
	In this section, we formally describe the proposed \swift, which introduces a trade-off between the communication loads and the number of active communication links among users and the server.  
	Consider a network consisting of one server and $N$ users, $U_1, U_2,\dots, U_N$, where up to $T$ of them are semi-honest which may collude with each other or with the server to gain some information about other users' local models. Furthermore, up to $D$ users may drop out, and their indices are denoted by $\mathcal{D}$.
  The operations need to be done in a finite field,  which is large enough to avoid hitting the boundary in the process of aggregation.  We choose a finite field $GF(p)$ denoted by $\mathbb{F}_p$, for some prime number $p$.
 Each user $n$ has a local model $\mathbf{W}_n\in \mathbb{F}_p^{L}$ and a set of random variables $\mathcal{Z}_n=\{\mathbf{Z}_{n,j},j\in[T] \}$ whose elements are chosen independently and uniformly at random from $\mathbb{F}_p^{\frac{L}{K}}$, for some parameter $K \in\mathbb{N}$ such that $(T+D+K)|N$.
	The server
	wants to recover the aggregated local models of the  surviving users, i.e., $\mathbf{W}=\sum_{n\in[N]\backslash\mathcal{D}}\mathbf{W}_{n}$, while the individual models remain private from semi-honest users and the server. 	To reach this goal,  \swift takes the following steps. 

	 {\bf 1) Grouping:} 	The set of $N$ users are arbitrarily partitioned into $\Gamma$ groups, each of size  $ \nu\triangleq D+T+K$, denoted by $\mathcal{G}_1, \mathcal{G}_2, \dots \mathcal{G}_{\Gamma}$. In each group, the users are labeled as user 1 to user $K+D+T$.  For simplicity, we refer to user $n$ based on its location in a group of users. 
	 Without loss of generality, we let $\gamma = \lfloor n/ \nu \rfloor$+1, and $t = n \mod \nu +1$, and place user $n$ on the $t$th location in group $\gamma$. Alternatively, we label user $n$ using $(\gamma, t)$.
	
	{\bf 2) Arranging the Groups:}  The groups are arranged on an arbitrary tree, which is called the \emph{aggregation tree} and denoted by $\mathscr{T}(\hat{\mathcal{V}},\hat{\mathcal{E}})$. The aggregation tree represents the flow of aggregation.
		 In the aggregation tree, $\hat{\mathcal{V}}$ is  a set of vertices representing groups and the server, $\hat{\mathcal{V}}=\{\mathcal{G}_1, \mathcal{G}_2, \dots, \mathcal{G}_{\Gamma},\text{Server}\}$ and $\hat{\mathcal{E}}$ is the set of edges representing the active connections between groups. In the aggregation tree, the server is the root and it has only one child which is the last group, $\mathcal{G}_{\Gamma}$. 
		 		For example, five different aggregation trees for an example with seven groups of users, $\mathcal{G}_1, \mathcal{G}_2,\dots,\mathcal{G}_7$ are shown in \figref{agg_tree}.

			
			For group $\mathcal{G}_{\gamma}$, located in one vertex of this tree, let $\mathcal{G}_{\gamma^+}$ represent the parent group of $\mathcal{G}_{\gamma}$, and  $\mathcal{F}_{(\gamma,\text{child})}$ be the set of indices of its children groups. In a rooted tree, a descendant of group $\gamma$ is any group $\gamma_1$ whose path from the root contains group $\gamma$, and group $\gamma_2$ is an ancestor of group $\gamma$ if and only if group ${\gamma}$ is a descendant of group ${\gamma_2}$. 
			Let us define $\mathcal{F}_{(\gamma,\text{desc})}$ as the set of indices of descendant groups of group $\gamma$, and $\mathcal{F}_{(\gamma,\text{anc})}$ as the set of indices of ancestor groups of group $\gamma$ (excluding the server).  For example, in \figref{agg_tree}~(d), $\mathcal{G}_{5^+}=\mathcal{G}_7$ is the parent group of $\mathcal{G}_5$ and $\mathcal{F}_{5,\text{child}}=\{1,2,3\}$. In addition, $\mathcal{G}_7$ has set $\mathcal{F}_{7,\text{desc}}=\{1,2,3,4,5,6\}$ and $\mathcal{G}_1$  has set $\mathcal{F}_{1,\text{anc}}=\{5,7\}$.

	{\bf 3) Partitioning the Local Models:}
	User $n$ partitions its local model into $K$ parts, denoted by 
	\begin{align}
	    \mathbf{W}_n=[\mathbf{W}_{n,1},\mathbf{W}_{n,2},\dots, \mathbf{W}_{n,K}],
	\end{align}
	where each part $\mathbf{W}_{n,k}, k\in[K]$ is a vector of size of $\frac{L}{K}$. 
	
	{\bf 4) Intra-Group Secret Sharing and Aggregation:}
	User $n\in[N]$ forms the following polynomial.
	\begin{align}\label{Fn}
	\mathbf{F}_n(x)=\sum\limits_{k=1}^{K}\mathbf{W}_{n,k}x^{k-1}+\sum\limits_{j=1}^{T}\mathbf{Z}_{n,j}x^{K+j-1}.
	\end{align}
		This polynomial function is designed such that the coefficient of  $x^k$ is the $(k+1)$th partition of the local model, for $k\in[0:K-1]$. Each user $n$ uses its polynomial $\mathbf{F}_n(.)$ to share its local model with other users.
		
		
		Let $\alpha_t\in \mathbb{F}_p$, for $t\in [\nu]$, be $\nu$ distinct non-zero constants. We assign   $\alpha_t$ to the $t$th user of all groups, i.e., users $(\gamma,t)$, $\gamma=1,\ldots, \Gamma$. 
		
		Within each group $\gamma$, each user $(\gamma,t)$  sends the evaluation of its polynomial function at $\alpha_{t'}$, i.e.,  $\mathbf{F}_{(\gamma,t)}(\alpha_{t'})$, to user $(\gamma,t')$, for $t'\in [\nu]$. 
		If a user $(\gamma,t)$ drops out and stays silent, $\mathbf{F}_{(\gamma,t)}(\alpha_{t'})$ is just presumed to be zero.
		
		Each user $(\gamma,t)$ calculates 
		\begin{align}
		\mathbf{Q}_{(\gamma,t)} = \sum_{t' \in [\nu]}  \mathbf{F}_{(\gamma, t')}(\alpha_{t}).
		\end{align}
		
		Note that in this phase, within each group, at most $\nu(\nu-1)$ communication take place, each of size $\frac{L}{K}$.

	
	\begin{figure}
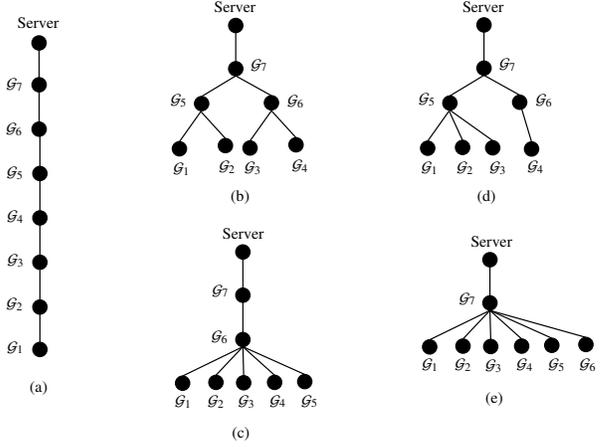


		\centering

		\tikzset{every picture/.style={line width=0.75pt}} 
		\scalebox{0.6}{

\tikzset{every picture/.style={line width=0.75pt}} 


			}
		\caption{Five examples for the aggregation tree for seven groups of users, $\mathcal{G}_1, \mathcal{G}_2, \dots,\mathcal{G}_7$, in \swift.  }
		\label{agg_tree}
	\end{figure}
		
		  {\bf 5) Inter-group Communication and Aggregation:}
		In this phase, user $t$ of group $\gamma$ calculates a message denoted by $\mathbf{S}_{(\gamma,t)}$ and  sends it to user $t$ of its parent group $\gamma^+$, for $\gamma\in[\Gamma-1]$. 
		
		
		Particularly, if user $(\gamma,t)$, $\gamma\in[\Gamma-1]$ does not have any children groups, i.e., $\mathcal{F}_{(\gamma,\text{child})}=\emptyset$, it sets
		\begin{align}\label{recursive1}
		   \mathbf{S}_{(\gamma,t)}=\mathbf{Q}_{(\gamma,t)}, 
		\end{align}
	    and sends $\mathbf{S}_{(\gamma,t)}$ to user $(\gamma^+,t)$.
		
		If $\mathcal{F}_{(\gamma,\text{child})}\ne\emptyset$, user $(\gamma,t)$ 
		calculates  $\mathbf{S}_{(\gamma,t)}$ as 
		\begin{align}\label{recursive2}
		\mathbf{S}_{(\gamma,t)}  = \mathbf{Q}_{(\gamma,t)}+ \sum\limits_{\gamma^-\in\mathcal{F}_{(\gamma,\text{child})}}\mathbf{S}_{(\gamma^-,t)},
		\end{align}
		upon receiving $\{\mathbf{S}_{(\gamma^-,t)},\gamma^-\in\mathcal{F}_{(\gamma,\text{child})}\}$. If user $(\gamma,t)$ does not receive $\mathbf{S}_{(\gamma^-,t)}$ from at least one group in $\mathcal{F}_{(\gamma,\text{child})}$, it also remains silent for the rest of the protocol. In this phase, at most $\nu(\Gamma-1)$ messages are communicated, each of size $\frac{L}{K}$.
		 Figure \ref{fig3} demonstrates the intra-group secret sharing and inter-group communication in a sequential aggregation tree.

			{\bf 6) Communication with the Server:}
		User $t$ of the last group, i.e., user $(\Gamma, t)$
		computes
		\begin{align}\label{recursive3}
		\mathbf{S}_{(\Gamma,t)}  =  \mathbf{Q}_{(\Gamma,t)}+ \sum\limits_{\gamma^-\in\mathcal{F}_{(\Gamma,\text{child})}}\mathbf{S}_{(\gamma^-,t)},
		\end{align}
		and sends it to the server, for $t\in [\nu]$.  
		
		 {\bf 7) Recovering the Result:}
		Having received the messages from a subset of users in $\mathcal{G}_{\Gamma}$ of size at least $T+K$, the server can recover the aggregated local models.
	
	 \begin{remark}
    In \eqref{Fn}, we use ramp sharing \cite{blakley1984security} to keep individual local models private, but any other arbitrary $T$-private MDS code can be used for $\mathbf{F}_n(x)$.  Ramp sharing allows us to reduce the size of shares and consequently achieve the cut-set outer bound within a diminishing gap.
	 \end{remark}
	 
	 \begin{remark}
The formation of the groups and the topology of the aggregation tree in \swift may be constrained by the actual connectivity in the network. 
	 \end{remark}
	\begin{figure*}
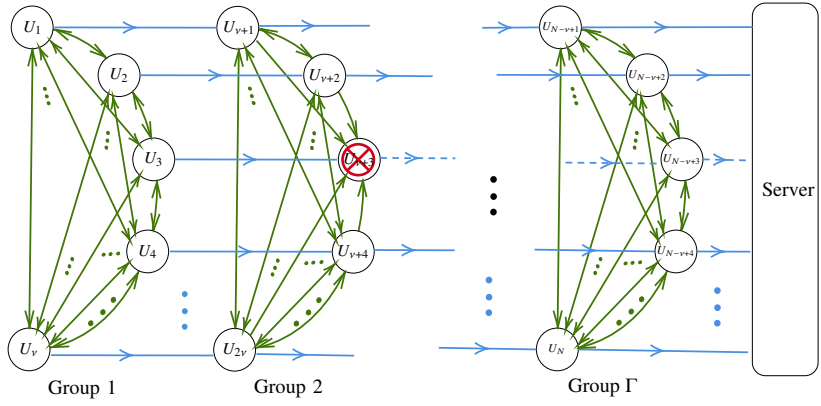

		\centering

		\tikzset{every picture/.style={line width=0.75pt}} 
		\scalebox{0.45}{

\tikzset{every picture/.style={line width=0.75pt}} 

	

		}
		\caption{An overview of the proposed setting in \swift with a sequential aggregation tree, where $\nu=T+D+K$. The intra-group and inter-group communication links are shown in green and blue, respectively.}
		\label{fig3}
	\end{figure*}

Next, we demonstrate the correctness of \swift in aggregating user models, by showing that the server can recover $\sum_{n\in[N]\backslash\mathcal{D}}\mathbf{W}_{n}$ from the messages received from group $\Gamma$. 
	Using the recursive \cref{recursive1,recursive2,recursive3}, $\mathbf{S}_{(\gamma,t)}(\alpha_t)$ is either a null message, or it is equal to
	\begin{align}\label{S_gamma}
	\mathbf{S}_{(\gamma,t)}(\alpha_t)=&\sum\limits_{k=1}^{K}{\alpha_t}^{k-1}\sum\limits_{\gamma'\in \mathcal{F}_{(\gamma,\text{desc})}}\sum\limits_{n\in \mathcal{G}_{\gamma'} \backslash \mathcal{D} }\hspace{-3mm}\mathbf{W}_{n,k}\\\nonumber
	&+\sum\limits_{j=1}^{T}{\alpha_t}^{K+j-1}\sum\limits_{\gamma'\in\mathcal{F}_{(\gamma,\text{desc})}}\sum\limits_{{n\in \mathcal{G}_{\gamma'} \backslash \mathcal{D} }
	}\mathbf{Z}_{n,j},
	\end{align}
	Thus if user $t$ in group $\Gamma$ sends a message to the server, it is equal to 	$\mathbf{S}_{(\Gamma,t)}(\alpha_t)$. From \eqref{S_gamma}, it is easy to see that $\mathbf{S}_{(\Gamma,t)}(\alpha_t)=\mathbf{F}(\alpha_t)$, where 
	\begin{align}\label{F}
	\mathbf{F}(x)=\sum\limits_{k=1}^{K}x^{k-1}\hspace{-2mm}\sum\limits_{n\in[N]\backslash\mathcal{D}}\hspace{-1mm}\mathbf{W}_{n,k}+ \sum\limits_{j=1}^{T}x^{K+j-1}\hspace{-2mm}\sum\limits_{n\in[N]\backslash\mathcal{D}}\hspace{-1mm}\mathbf{Z}_{n,j}.
	\end{align}
	$\mathbf{F}(x)$ is a polynomial of degree $K+T-1$.  Thus if the server receives at least $T+K$ messages from the last group, it can use Lagrange interpolation to recover  $\mathbf{F}(x)$ and
	\begin{align*}
	   \sum_{n\in[N]\backslash\mathcal{D}}\mathbf{W}_{n}=\bigg[\hspace{-1mm}\sum_{n\in[N]\backslash\mathcal{D}}\hspace{-1mm}\mathbf{W}_{n,1}, \hspace{-1mm}\sum_{n\in[N]\backslash\mathcal{D}}\hspace{-1mm}\mathbf{W}_{n,2},\dots,\hspace{-1mm}\sum_{n\in[N]\backslash\mathcal{D}}\hspace{-1mm}\mathbf{W}_{n,K}\bigg]^T.
	\end{align*}
	
	Recall that, in \swift, for any user $(\gamma,t)$ in $\mathcal{D}$, all messages $\mathbf{S}_{(\gamma',t)}(\alpha_t)$, $\gamma'\in\{\gamma,\mathcal{F}_{(\gamma,\text{anc})}\}$ is null. In particular, for any user $(\gamma,t)$ in $\mathcal{D}$,  $\mathbf{S}_{(\Gamma,t)}(\alpha_t)$ is null. Thus at most $D$ users in the last group send null messages to the server. Since the size of each group is $D+T+K$, the server receives at least $T+K$ values  $\mathbf{S}_{(\Gamma,t)}(\alpha_t)$ for distinct $\alpha_t$, and thus can correctly recover $\mathbf{F}(x)$. 
	
\section{Theoretical Analysis}	In this section, we analyze the communication and computation complexities of the proposed \swift scheme, and provide an information-theoretic proof on its privacy guarantee in protecting local model information from colluding users and the server.
	
	\subsection{Communication Loads and Number of Active Connections}\label{proof_comm}
	According to \eqref{F}, the total number of symbols that are needed to be received by the server is $T+K$, each of size $\frac{L}{K}$. Thus, the normalized server communication load in \swift is $	 R_{\text{server}}^{(L)}=1+\frac{T}{K}.$
	
 In each group, at most $\nu(\nu-1)$ symbols are sent by the members, each of size $\frac{L}{K}$, and there are $\frac{N}{\nu}$ groups. In addition, at most $\nu$ symbols are sent from a child group to a parent group. Thus, the normalized per-user communication load in \swift is upper-bounded as $R_{\text{user}}^{(L)}\le(1+\frac{T+D}{K})$, and the number of active communication links in the communication graph of \swift, i.e., $|\mathcal{E}_{\swift}|$, is upper bounded by $\frac{N}{\nu}\left(\frac{\nu(\nu-1)}{2}+\nu\right)=\frac{N}{2}(K+T+D+1)$.
 

It is clear that through adjusting the value of $K$ between $1$ and $N-D-T$, we can have a trade-off between the communication loads and the number of active connections in the network. One choice of interest is $K=N-D-T$ as the largest value for $K$, which leads to Theorem~\ref{theorem} and minimizes the communication loads. It is also possible to minimize the number of network connections by choosing $K=1$, as proposed in \cite{jahani2022swiftagg}.

While the communication loads and the number of active communication links are determined by the value of $K$ using \swift, the actual communication delay of the model aggregation also depends on the topology of the adopted aggregation tree.
	 Assume that each communication link between two groups has maximum delay $\delta_{\text{Inter}}$. In the case of more than one group, the  total delay of each training iteration, $\Delta$, can be minimized if
$\mathcal{F}_{(\Gamma,\text{child})}=\{\gamma_1,\gamma_2,\dots,\Gamma-1\}$.
Thus, in a synchronous system $\Delta=2\delta_{\text{Inter}}+\delta_{\text{Intra}}$, where $\delta_{\text{Intra}}$ is the delay for intra-group communications. This case is shown for an example consisting of 7 groups in \figref{agg_tree}(e). In contrast, the maximum delay occurs when the groups are located sequentially in the aggregation tree which is shown for the example in \figref{agg_tree}(a), and $\Delta=\Gamma\delta_{\text{Inter}}+\delta_{\text{Intra}}$.

	\subsection{Computation Loads}
	In this subsection, we analyze the computation loads at the server and at each user with respect to the maximum number semi-honest users, $T$, maximum number of dropouts, $D$, and size of the local models, $L$, and the design parameter $1\le K\le N-T-D$ that is the number of sub-vectors each local model is broken into.
	
	{\bf Computational Complexity at the Server:} In \swift, to recover the aggregation of local models, the server needs to receive $T+K$ outcomes of the users in the last group, each of size $\frac{L}{K}$. Therefore, the computation performed by the server includes interpolation of a polynomial of degree $T+K-1$, where the coefficients are vectors of size $\frac{L}{K}$. The complexity of interpolation of a polynomial of degree $d$ is $\mathcal{O}(d\log^2d)$, when the field supports FFT \cite{kedlaya2011fast}. Thus, the computational complexity at the server is $\mathcal{O}\big((1+\frac{T-1}{K})L\log^2(K+T-1)\big)$.
	
	{\bf Computational Complexity at the User:} In \swift, each user performs the following two operations:
	\begin{enumerate}
	    \item It evaluates one polynomial in $\nu=K+T+D$ distinct values. If the field supports FFT, evaluation of a
polynomial function of degree $d$ has a computational complexity of $\mathcal{O}(\log^2d)$ \cite{kedlaya2011fast}.
	    In \swift, we need to compute one polynomial of degree $T+K-1$, where the coefficients are vectors of size $\frac{L}{K}$, at $\nu$ points. Thus, this step requires a complexity of $\mathcal{O}\big((T+K+D)\frac{L}{K}\log^2{(T+K-1)}\big)$.
	    \item On average each user computes the summation of $\nu+1$ vectors of size $\frac{L}{K}$ which has a complexity of $\mathcal{O}\big((T+K+D) \frac{L}{K} \big)$. 
	\end{enumerate}
	Therefore, the computation load on each user is $\mathcal{O}\big( (1+\frac{T+D}{K}){L}(1+\log^2(T+K-1))\big)$.
\subsection{Proof of Theorem \ref{th-server}}\label{proof-th-server}
In \swift, we choose a finite field $\mathbb{F}_p$, for some prime number $p$, where $N(\ell-1) < p \leq 2N (\ell-1)$. According to Bertrand's postulate \cite{aigner1999proofs} such prime number exists for $N(\ell-1) \ge 1$. In terms of number of required bits, the server communication load achieved by \swift is as follows.
\begin{align}\label{server-swift}
    R_{\text{server}} &\leq (1+\frac{T}{N-T-D})\lceil \log_2(p) \rceil \\\nonumber
    &\leq (1+\frac{T}{N-T-D}) \bigg( \log_2\big((\ell-1) N\big) +1\bigg).
\end{align}
\textbf{Cut-set lower bound:} We know that the local models are from some joint distribution 
     $\mathbf{W}_1,\mathbf{W}_2,\dots,\mathbf{W}_N\sim P_{\mathbf{W}_1,\mathbf{W}_2,\dots,\mathbf{W}_N}(\mathbf{W}_1,\mathbf{W}_2,\dots,\mathbf{W}_N)$. However,  neither the users nor the server know about this joint distribution.   On the other hand, due to the privacy conditions, the users cannot, implicitly or explicitly,  learn the join distribution and adapt their transmission schemes to the joint distribution of the local models. 
     For any joint distribution, the server should be able to recover the aggregation of local models using the received messages from the users. For the cut separating the users from the server,  the worst joint distribution  occurs when the distribution of the aggregation of local models on the server becomes uniform.  Therefore, in terms of the number of required bits, the cut-set lower bound for the server communication load is derived as follows.
      \begin{align}\label{server-cut}
              R_{\text{server}}^{(L)}\ge \log_2\big({(\ell-1) N}+1\big).
          \end{align}
          Comparing \eqref{server-cut} with \eqref{server-swift} shows that for \texttt{SwiftAgg+} we have
\begin{align*}
    \frac{R_{\text{server}}^{(\text{Achievable})}}{R_{\text{server}}^{(\text{Lower})}} &\leq (1+ \frac{T}{N-T-D})\frac{\log_2((\ell-1)N)}{\log_2((\ell-1)N+1)} \\
    &+(1 +\frac{T}{N-T-D})\frac{ 1}{\log_2{\big((\ell-1) N+1\big)}}.
\end{align*}
Therefore, if $T=o(N)$ and $N-D=\mathcal{O}(N)$ we have
\begin{align*}
    \frac{R_{\text{server}}^{(\text{Achievable})}}{R_{\text{server}}^{(\text{Lower})}}-1    &\approx \mathcal{O}\bigg(\frac{ 1}{\log_2{\big((\ell-1) N\big)}}\bigg),
\end{align*}
where the server communication load in \texttt{SwiftAgg+} is within factor 1 of the cut-set lower bound.
\subsection{Proof of Theorem \ref{th-user}}\label{proof-th-user}
In terms of number of required bits, the per-user communication load achieved by \swift is as follows.
\begin{align}\label{user-swift}
    R_{\text{user}}&\le (1+\frac{T+D}{N-T-D})\lceil \log_2(p) \rceil\\ \nonumber
    &\leq (1+\frac{T+D}{N-T-D})\bigg( \log_2\big((\ell-1) N\big) +1\bigg).
\end{align}
\textbf{Cut-set lower bound:} For uniform and independent local models, and for the cut separating each user from the rest of the network, 
the cut-set lower bound is  as follows.
\begin{align}\label{user-cut}
              R_{\text{user}}^{(L)}\ge \log_2{\ell}.
\end{align}
Comparing \eqref{user-cut} with \eqref{user-swift} shows that if $T=o(N)$ and $D=o(N)$ we have 
\begin{align*}
      \frac{R_{\text{user}}^{(\text{Achievable})}}{R_{\text{user}}^{(\text{Lower})}}- \log_{\ell}{\ell N}   &\approx \mathcal{O}\bigg(\frac{\log_{\ell}{N}}{N}\bigg),
 \end{align*}
 where the per-user communication load in  \texttt{SwiftAgg+} is within factor $\log_{\ell}{\ell N} $ of the cut-set lower bound.
	\subsection{Proof of Privacy}\label{basic_privacy}

In this section, we prove that \swift satisfies the privacy constraint in \eqref{serve-prv}. The privacy must be guaranteed even if the server colludes with any set $\mathcal{T}\subset [N]$ of at most $T$ semi-honest users which can distribute arbitrary across the groups. 
At a high level, we expand the mutual information in \eqref{serve-prv} over the groups containing the semi-honest users, from the leaf to the root of the aggregation tree, and show that model privacy will be preserved at each expansion step.

In \swift, the \emph{Intra-group Secret Sharing} is inspired by ramp secret sharing scheme \cite{blakley1984security}. Ramp secret sharing scheme is proposed to reduce the size of shares in Shamir's secret sharing \cite{shamir1979share}. It proposes a trade-off between security and size of the shares. In ramp sharing scheme with $T$-privacy, 
no information can be leaked from any $T$ or less shares. This directly leads to the following corollary for \swift.

	\begin{corollary}\label{coro4}
		Assume that user $U_n$ is denoted by $(\gamma,t)$.  In \swift, the local model of $U_n$, is shared using polynomial function $\mathbf{F}_n(x)$ in \eqref{Fn}. In other words, $\mathbf{F}_{(\gamma,t)}(\alpha_{t'})$ for $t'\in[\nu]\backslash \{t\}$ are delivered to user $(\gamma,t')$.
		According to \eqref{Fn} and directly from the privacy guarantee in ramp sharing, we have
		$I(\mathbf{W}_{n};\{ \mathbf{F}_n(\alpha_{t'}).t'\in\mathcal{T}\})=0.$
	\end{corollary}

Let the random part of inter-group message $\mathbf{S}_{(\gamma,t)}$ of user $(\gamma,t)$, consisting of the random noises of non-dropped and honest descendant, be denoted by $\tilde{\mathbf{Z}}_{(\gamma,t)}$, i.e.,
	\begin{align}\label{z_tild}
	    \tilde{\mathbf{Z}}_{(\gamma,t)}\triangleq\sum_{j=1}^{T}{\alpha_t}^{K+j-1}\sum_{\gamma'\in\{\gamma,\mathcal{F}_{(\gamma,\text{desc})}\}}\sum_{{n\in \mathcal{G}_{\gamma'} \backslash \{\mathcal{D}\cup \mathcal{T}\} }}\mathbf{Z}_{n,j}.
	\end{align}
	

We show in the following lemma, that the random noise in the message sent from a child to its parent is independent of the random noise in the message sent from that parent to the ancestor group. 

	\begin{lemma}\label{lemma1}
For all $\gamma\in[\Gamma]$ and $t\in[\nu]$, and group $\gamma^+$ as the parent group of group $\gamma$, we have 
 $   I\big(\tilde{\mathbf{Z}}_{(\gamma,t)};\tilde{\mathbf{Z}}_{(\gamma^+,t)}\big)=0.$
\end{lemma}
	
\begin{proof}
	     Let us define
	     \begin{align}\label{z'}
	         {\mathbf{Z}'}_{n}^{(t)}\triangleq\sum_{j=1}^{T}\mathbf{Z}_{n,j}\alpha_{t}^{K+j-1},
	     \end{align}
	     for $n\in[N]$ and $t\in[\nu]$. In each group, there are $\nu=T+D+K$ users each of which uses $T$ random vectors chosen uniformly and independently from $\mathbb{F}_p^{\frac{L}{K}}$ in its shares. In addition, we have
	   $ \tilde{\mathbf{Z}}_{(\gamma^+,t)}=\sum_{\gamma\in\mathcal{F}_{(\gamma^+,\text{child})}}\tilde{\mathbf{Z}}_{(\gamma,t)}+\sum_{i\in\mathcal{G}_{\gamma^+}\backslash\{\mathcal{D}\cup \mathcal{T}\}}{\mathbf{Z}'}_i^{(t)}$. Since $\mathcal{G}_{\gamma^+}\backslash\{\mathcal{D}\cup \mathcal{T}\} \neq \emptyset$, and 
	    the random vectors $\{{\mathbf{Z}'}_i^{(t)}\}_{i\in\mathcal{G}_{\gamma^+}\backslash\{\mathcal{D}\cup \mathcal{T}\}}$ are i.i.d., we have that $ I\big(\tilde{\mathbf{Z}}_{(\gamma,t)};\tilde{\mathbf{Z}}_{(\gamma^+,t)}\big)=0$ for $\gamma\in\mathcal{F}_{(\gamma^+,\text{desc})}$.
	\end{proof}



	Assume that the semi-honest users are denoted by $\tilde{U}_1, \tilde{U}_2,\dots,\tilde{U}_T$. We denote the indices of theses semi-honest users as $(\gamma_1,t_1),(\gamma_2,t_2),\dots,(\gamma_T,t_T)$ respectively.
	 We also denote the set of indices of honest users in group $\gamma$, $\gamma\in[\Gamma]$, by $\mathcal{H}_{\gamma}\triangleq\{n:U_n\in\mathcal{G}_{\gamma}\backslash\{\mathcal{T}\cup \mathcal{D}\}\}$.

	    Let us define the set of messages which are received by $\tilde{U}_{i}$ by $\mathcal{M}_{\tilde{U}_{i}}$ which consists of two kinds of messages. Particularly, $\mathcal{M}_{\tilde{U}_{i}}=\big\{\{\mathbf{F}_{(n,t_i)},n\in\mathcal{H}_{\gamma_i}\},\{\mathbf{S}_{(\gamma'_i,t_i)},\gamma'_i\in\mathcal{F}_{(\gamma_i,\text{child})}\}  \big\}$,
	    where
 	    $ \{\mathbf{F}_{(n,t_i)},n\in\mathcal{H}_{\gamma_i}\}
 	    =\{\mathbf{W}_n+{\mathbf{Z}'}_{n}^{(t_i)}, n\in\mathcal{H}_{\gamma_i}\}$,
 	is a set of intra-group messages, and $ \mathbf{S}_{(\gamma'_i,t_i)}=\sum_{\gamma'\in\{\gamma_i',\mathcal{F}_{(\gamma_i',\text{desc})}\}}\sum_{m\in \mathcal{H}_{\gamma'}} \mathbf{W}_m+\tilde{\mathbf{Z}}_{(\gamma'_i,t_i)},$
 	  is the message received from the child group $\gamma_i'$. 
 	  
 According to the ramp secret sharing scheme, definitions in \eqref{z_tild},\eqref{z'},  and the fact that the random vectors are chosen uniformly and independently at random from $\mathbb{F}_p^{\frac{L}{K}}$, we can easily prove the following lemmas.
 
\begin{lemma}\label{lemma2}
For each user $(\gamma_i,t_i)$ in \swift we have $I\big( \{{\tilde{\mathbf{Z}}}_{(\gamma'_i,t_i)},\gamma'_i\in\mathcal{F}_{(\gamma_i,\text{child})}\};\{{\mathbf{Z}'}_n^{(t_i)},n\in\mathcal{H}_{\gamma_i}\} \big)=0$, for all $t_i\in[\nu]$ and $\gamma_i\in[\Gamma]$.
\end{lemma}
	  
	 \begin{lemma}\label{lemma3}
	 Consider user $(\gamma_i,t_i)$ and user $(\tilde{\gamma}_i,\tilde{t}_i)$. Then, $I\big( \{{\mathbf{Z}'}_n^{(t_i)},n\in\mathcal{H}_{\gamma_i}\}; \{{\mathbf{Z}'}_{{n}}^{(\tilde{t}_i)},{n}\in\mathcal{H}_{\tilde{\gamma}_i}\}\big)=0$ for $\gamma_i\neq \tilde{\gamma}_i$, $t_i,\tilde{t}_i\in[\nu]$.
	 \end{lemma}  

\begin{lemma}\label{lemma4}
 For any $\gamma_i\in[\Gamma]$ consider user $n\in\mathcal{H}_{\gamma_i}$. Then 
 for all $t_i\in[\nu]$, $I\big( {\mathbf{Z}'}_n^{(t_i)};\{{\mathbf{Z}'}_n^{(t)},t\in\mathcal{T}'\}\big)=0$, where $\mathcal{T}'\subset[\nu]\backslash \{t_i\}$, and $|\mathcal{T}'|\le T-1$. Similarly, $I\big( \tilde{\mathbf{Z}}_{(\gamma_i,t_i)};\{\tilde{\mathbf{Z}}_{(\gamma_i,t)},t\in\mathcal{T}'\}\big)=0$.
\end{lemma}

Let us define $\mathcal{W}_{N\backslash\mathcal{T}}\triangleq\{\mathbf{W}_n,{n\in[N]}\backslash{\mathcal{T}}\}$,  $\mathcal{K}_{N,\mathcal{T}}\triangleq\{ \{\mathbf{W}_k,\mathcal{Z}_k,{k\in\mathcal{T}}\},\sum_{n\in[N]\backslash\{\mathcal{D}\cup\mathcal{T}\}}{\mathbf{W}_n}\}$, and $\mathcal{M}_{\mathcal{T}}~\triangleq~\bigcup_{i\in[T]}\mathcal{M}_{\tilde{U}_i}$. In addition, $\mathcal{S}_{\Gamma}~\triangleq~\{\mathbf{S}_{(\Gamma,t)},{t\in[\nu]\backslash\mathcal{D}}\}$ represents the set of messages that the server receives from users in group $\Gamma$.
According to privacy constraint in \eqref{serve-prv}, we have to show that
	\begin{align}\label{25}
I\big(\mathcal{W}_{N\backslash\mathcal{T}};\mathcal{M}_{\mathcal{T}},\mathcal{S}_{\Gamma}\big|\hspace{-1mm}\sum\limits_{n\in[N]\backslash\{\mathcal{D}\cup\mathcal{T}\}}\hspace{-6mm}{\mathbf{W}_n},\{\mathbf{W}_k,\mathcal{Z}_k,{k\in\mathcal{T}}\} \big)=0.
		\end{align}
	\begin{lemma}\label{lemma7}
	Let $(\gamma_1,t_1),(\gamma_2,t_2),\dots,(\gamma_T,t_T)$ be $T$ semi-honest users, where $\gamma_i\in[\Gamma]$ and $t_i\in[\nu]$ for $i\in[T]$. Then, for $i\in[T]$ we have
$I\big(\mathcal{W}_{N\backslash\mathcal{T}};\mathcal{M}_{\tilde{U}_i}\big|\mathcal{K}_{N,\mathcal{T}},\{\mathcal{M}_{\tilde{U}_j},j\in\mathcal{J}_i\} \big)=0,$
where $\mathcal{J}_i=\{j: \gamma_j\in\mathcal{F}_{(\gamma_i,\text{desc})}, \text{ for } j\in[T]\}$.
\begin{proof}
	    Consider the semi-honest user in group $\gamma_i$, $i\in[T]$. If $\mathcal{J}_i=\emptyset$, then we have 
 	    \begin{align*}
	        I&\big(\mathcal{W}_{N\backslash\mathcal{T}};\mathcal{M}_{\tilde{U}_i}\big|\mathcal{K}_{N,\mathcal{T}} \big)\\ =&I\big(\mathcal{W}_{N\backslash\mathcal{T}};\{\mathbf{F}_{(n,t_i)}, n\in\mathcal{H}_{\gamma_i}\},\{\mathbf{S}_{(\gamma'_i,t_i)},\gamma'_i\in\mathcal{F}_{(\gamma_i,\text{child})}\}\big|\mathcal{K}_{N,\mathcal{T}} \big)\\
	        =&H\big(\{\mathbf{F}_{(n,t_i)}, n\in\mathcal{H}_{\gamma_i}\}, \{\mathbf{S}_{(\gamma'_i,t_i)},\gamma'_i\in\mathcal{F}_{(\gamma_i,\text{child})}\}\big|\mathcal{K}_{N,\mathcal{T}} \big)\\
	        &-H\big(\{\mathbf{F}_{(n,t_i)}, n\in\mathcal{H}_{\gamma_i}\},\{\mathbf{S}_{(\gamma'_i,t_i)},\gamma'_i\in\mathcal{F}_{(\gamma_i,\text{child})}\}\big|\mathcal{K}_{N,\mathcal{T}},\mathcal{W}_N\big)\\
	        \leq &H\big(\{\mathbf{F}_{(n,t_i)}, n\in\mathcal{H}_{\gamma_i}\},\{\mathbf{S}_{(\gamma'_i,t_i)},\gamma'_i\in\mathcal{F}_{(\gamma_i,\text{child})}\}\big)\\
	   & -H\big(\{{\mathbf{Z}'}_n^{(t_i)},n\in\mathcal{H}_{\gamma_i}\},\{\tilde{\mathbf{Z}}_{(\gamma'_i,t_i)},\gamma'_i\in\mathcal{F}_{(\gamma_i,\text{child})}\}\big)\le 0.
	    \end{align*}
	    The last term follows from the fact that both terms 
	    have the same size and uniform variables maximize the entropy.  Therefore, $I\big(\mathcal{W}_{N\backslash\mathcal{T}};\mathcal{M}_{\tilde{U}_i}\big|\mathcal{K}_{N,\mathcal{T}} \big)=0$.

 \noindent If $\mathcal{J}_i\ne\emptyset$, then we have 
	    \begin{align}\nonumber
	        I\big(&\mathcal{W}_{N\backslash\mathcal{T}};\mathcal{M}_{\tilde{U}_i}\big|\mathcal{K}_{N,\mathcal{T}},\{\mathcal{M}_{\tilde{U}_j},j\in\mathcal{J}_i\} \big)\\\nonumber
	        =I\bigg(&\mathcal{W}_{N\backslash\mathcal{T}};\{\mathbf{F}_{(n,t_i)}, n\in\mathcal{H}_{\gamma_i}\},\{\mathbf{S}_{(\gamma'_i,t_i)},\gamma'_i\in\mathcal{F}_{(\gamma_i,\text{child})}\}\bigg|
	        \mathcal{K}_{N,\mathcal{T}},\\\nonumber
	        &\big\{\{\mathbf{F}_{(n,t_{j})},n\in\mathcal{H}_{\gamma_j}\},\{\mathbf{S}_{(\gamma'_i,t_j)},\gamma'_j\in\mathcal{F}_{(\gamma_j,\text{child})}\},j\in\mathcal{J}_i\big\} \bigg)\\\nonumber
	        =I\bigg(&\mathcal{W}_{N\backslash\mathcal{T}};\{\mathbf{F}_{(n,t_i)}, n\in\mathcal{H}_{\gamma_i}\}\bigg|
	        \mathcal{K}_{N,\mathcal{T}},
	         \big\{\{\mathbf{F}_{(n,t_{j})},n\in\mathcal{H}_{\gamma_j}\},\\\nonumber
	         &\{\mathbf{S}_{(\gamma'_i,t_j)},\gamma'_j\in\mathcal{F}_{(\gamma_j,\text{child})}\},j\in\mathcal{J}_i\big\}  \bigg)\\\nonumber
	        +I\bigg(&\mathcal{W}_{N\backslash\mathcal{T}};\{\mathbf{S}_{(\gamma'_i,t_i)},\gamma'_i\in\mathcal{F}_{(\gamma_i,\text{child})}\}\bigg|
	        \mathcal{K}_{N,\mathcal{T}},\{\mathbf{F}_{(n,t_i)}, n\in\mathcal{H}_{\gamma_i}\},\\\label{eq_I}
	        & \big\{\{\mathbf{F}_{(n,t_{j})},n\in\mathcal{H}_{\gamma_j}\},\{\mathbf{S}_{(\gamma'_i,t_j)},\gamma'_j\in\mathcal{F}_{(\gamma_j,\text{child})}\},j\in\mathcal{J}_i\big\} \bigg).
	        \end{align}
	        Using the definition of mutual information, \eqref{eq_I} can be written as 
	         \begin{align}\nonumber\label{eq18}
	       H\bigg(&\{\mathbf{F}_{(n,t_i)}, n\in\mathcal{H}_{\gamma_i}\}\bigg|
	        \mathcal{K}_{N,\mathcal{T}},
	         \big\{\{\mathbf{F}_{(n,t_{j})},n\in\mathcal{H}_{\gamma_j}\},
	         \\\nonumber
	         &\{\mathbf{S}_{(\gamma'_i,t_j)},\gamma'_j\in\mathcal{F}_{(\gamma_j,\text{child})}\},j\in\mathcal{J}_i\big\}\bigg)\\\nonumber
	        - H\bigg(&\{\mathbf{F}_{(n,t_i)}, n\in\mathcal{H}_{\gamma_i}\}\bigg|
	        \mathcal{W}_{N\backslash\mathcal{T}},\mathcal{K}_{N,\mathcal{T}},
	         \big\{\{\mathbf{F}_{(n,t_{j})},n\in\mathcal{H}_{\gamma_j}\},\\\nonumber
	         &\{\mathbf{S}_{(\gamma'_i,t_j)},\gamma'_j\in\mathcal{F}_{(\gamma_j,\text{child})}\},j\in\mathcal{J}_i\big\}\bigg)\\\nonumber
	        +H\bigg(&\{\mathbf{S}_{(\gamma'_i,t_i)},\gamma'_i\in\mathcal{F}_{(\gamma_i,\text{child})}\}\bigg|
	        \mathcal{K}_{N,\mathcal{T}},\{\mathbf{F}_{(n,t_i)}, n\in\mathcal{H}_{\gamma_i}\},\\\nonumber
	         &\big\{\{\mathbf{F}_{(n,t_{j})},n\in\mathcal{H}_{\gamma_j}\},\{\mathbf{S}_{(\gamma'_i,t_j)},\gamma'_j\in\mathcal{F}_{(\gamma_j,\text{child})}\},j\in\mathcal{J}_i\big\}\bigg)\\\nonumber
	        -H\bigg(&\{\mathbf{S}_{(\gamma'_i,t_i)},\gamma'_i\in\mathcal{F}_{(\gamma_i,\text{child})}\}\bigg|
	        \mathcal{W}_{N\backslash\mathcal{T}},\mathcal{K}_{N,\mathcal{T}},\{\mathbf{F}_{(n,t_i)}, n\in\mathcal{H}_{\gamma_i}\},\\\nonumber
	         &\big\{\{\mathbf{F}_{(n,t_{j})},n\in\mathcal{H}_{\gamma_j}\},\{\mathbf{S}_{(\gamma'_i,t_j)},\gamma'_j\in\mathcal{F}_{(\gamma_j,\text{child})}\},j\in\mathcal{J}_i\big\}\bigg)\\\nonumber
	        \stackrel{\text{(a)}}{\leq }H\bigg(&\{\mathbf{F}_{(n,t_i)}, n\in\mathcal{H}_{\gamma_i}\}\bigg)
	       +H\bigg(\{\mathbf{S}_{(\gamma'_i,t_i)},\gamma'_i\in\mathcal{F}_{(\gamma_i,\text{child})}\}\bigg)\\\nonumber
	       	-H\bigg(&\{{\mathbf{Z}'}_n^{(t_i)}, n\in\mathcal{H}_{\gamma_i}\}\bigg|
	        \mathcal{W}_{N},\mathcal{K}_{N,\mathcal{T}},
	         \big\{\{{\mathbf{Z}'}_n^{(t_j)},n\in\mathcal{H}_{\gamma_j}\},\\\nonumber
	         &\{\tilde{\mathbf{Z}}_{(\gamma'_i,t_j)},\gamma'_j\in\mathcal{F}_{(\gamma_j,\text{child})}\},j\in\mathcal{J}_i\big\}\bigg)\\\nonumber
	       -H\bigg(&\{\tilde{\mathbf{Z}}_{(\gamma'_i,t_i)},\gamma'_i\in\mathcal{F}_{(\gamma_i,\text{child})}\}\bigg|
	        \mathcal{W}_{N},\mathcal{K}_{N,\mathcal{T}},\{{\mathbf{Z}'}_n^{(t_i)}, n\in\mathcal{H}_{\gamma_i}\},\\
	        & \big\{\{{\mathbf{Z}'}_n^{(t_{j})},n\in\mathcal{H}_{\gamma_j}\},\{\tilde{\mathbf{Z}}_{(\gamma'_i,t_j)},\gamma'_j\in\mathcal{F}_{(\gamma_j,\text{child})}\},j\in\mathcal{J}_i\big\}\bigg),
	         \end{align}
	          where in (a) the first and the second terms follow from the fact that $H(X|Y)\le H(X)$. 
	          
	          Now we show that 
	 \begin{align}\nonumber\label{eq19}
	            I\bigg(&\{{\mathbf{Z}'}_n^{(t_i)}, n\in\mathcal{H}_{\gamma_i}\};
	        \mathcal{W}_{N},\mathcal{K}_{N,\mathcal{T}},
	         \big\{\{{\mathbf{Z}'}_n^{(t_j)},n\in\mathcal{H}_{\gamma_j}\},\\
	         &\{\tilde{\mathbf{Z}}_{(\gamma'_i,t_j)},\gamma'_j\in\mathcal{F}_{(\gamma_j,\text{child})}\},j\in\mathcal{J}_i\big\}\bigg)=0.
	         \end{align}
	         From definition of mutual information, we have
	         \begin{align*}
	              I\bigg(&\{{\mathbf{Z}'}_n^{(t_i)}, n\in\mathcal{H}_{\gamma_i}\};
	        \mathcal{W}_{N},\mathcal{K}_{N,\mathcal{T}},
	         \big\{\{{\mathbf{Z}'}_n^{(t_j)},n\in\mathcal{H}_{\gamma_j}\},\\
	         &\{\tilde{\mathbf{Z}}_{(\gamma'_i,t_j)},\gamma'_j\in\mathcal{F}_{(\gamma_j,\text{child})}\},j\in\mathcal{J}_i\big\}\bigg)\\
	         =H\bigg(&  \mathcal{W}_{N},\mathcal{K}_{N,\mathcal{T}},
	         \big\{\{{\mathbf{Z}'}_n^{(t_j)},n\in\mathcal{H}_{\gamma_j}\},
	         \{\tilde{\mathbf{Z}}_{(\gamma'_i,t_j)},\gamma'_j\in\mathcal{F}_{(\gamma_j,\text{child})}\},\\
	         &j\in\mathcal{J}_i\big\}\bigg)-H\bigg(  \mathcal{W}_{N},\mathcal{K}_{N,\mathcal{T}},
	         \big\{\{{\mathbf{Z}'}_n^{(t_j)},n\in\mathcal{H}_{\gamma_j}\},\\
	         &\{\tilde{\mathbf{Z}}_{(\gamma'_i,t_j)},\gamma'_j\in\mathcal{F}_{(\gamma_j,\text{child})}\},j\in\mathcal{J}_i\big\}\bigg|\{{\mathbf{Z}'}_n^{(t_i)}, n\in\mathcal{H}_{\gamma_i}\} \bigg)=0,
	         \end{align*}
	         where the last equality holds due to Lemma \ref{lemma2}, Lemma~\ref{lemma3}, Lemma~\ref{lemma4}, and the independence of local models and random vectors.
	           \begin{lemma}\label{lemma6}
	        For $\gamma_i\in[\Gamma]$, consider ${{\mathcal{Z}}'}_{{{\tilde{\mathcal{H}}}_i}}=\{{\mathbf{Z}'}_n^{(t)},n\in{\tilde{\mathcal{H}}}_i,t\in\tilde{\mathcal{T}}\}$,  where ${\tilde{\mathcal{H}}}_i$ is a subset of $\{\mathcal{H}_j,j\in\mathcal{F}_{(\gamma_i,\text{desc})}\}$ of size up to $T-1$, and $|\tilde{\mathcal{T}}|\le T-1$. Then, $I\big({\tilde{\mathbf{Z}}}_{(\gamma_i,t)};{{\mathcal{Z}}'}_{{{\tilde{\mathcal{H}}}_i}}\big)=0$.
	        \begin{proof}
	   We can consider two cases: (I) If there is at least one group like $\gamma'_i\in\mathcal{F}_{(\gamma_i,\text{desc})}$ such that $\mathcal{H}_{\gamma'_i}\notin{\tilde{\mathcal{H}}}_i$ then we can conclude that
	      $I\big({\tilde{\mathbf{Z}}}_{(\gamma_i,t)};{{\mathcal{Z}}'}_{{{\tilde{\mathcal{H}}}_i}}\big)=0$.
	     The reason is that there is a non-empty set of honest and non-dropped users that  ${\tilde{\mathbf{Z}}}_{(\gamma_i,t)}$  includes a summation of their i.i.d. random vectors, (II) If ${\tilde{\mathcal{H}}}_i=\{\mathcal{H}_j,j\in\mathcal{F}_{(\gamma_i,\text{desc})}\}$,  then based on ramp secret sharing we have  $I\big({\tilde{\mathbf{Z}}}_{(\gamma_i,t)};{{\mathcal{Z}}'}_{{{\tilde{\mathcal{H}}}_i}}\big)=0$.
	        \end{proof}
	    \end{lemma}
	         
	        Similar to \eqref{eq19}, using Lemma~\ref{lemma1}, Lemma~\ref{lemma2}, Lemma~\ref{lemma4}, and  Lemma~\ref{lemma6} we can proof 
	         \begin{align}\nonumber\label{eq20}
	       I\bigg(&\{\tilde{\mathbf{Z}}_{(\gamma'_i,t_i)},\gamma'_i\in\mathcal{F}_{(\gamma_i,\text{child})}\};
	        \mathcal{W}_{N},\mathcal{K}_{N,\mathcal{T}},\{{\mathbf{Z}'}_n^{(t_i)}, n\in\mathcal{H}_{\gamma_i}\},\\
	        & \big\{\{{\mathbf{Z}'}_n^{(t_{j})},n\in\mathcal{H}_{\gamma_j}\},\{\tilde{\mathbf{Z}}_{(\gamma'_i,t_j)},\gamma'_j\in\mathcal{F}_{(\gamma_j,\text{child})}\},j\in\mathcal{J}_i\big\}\bigg)=0.
	         \end{align}
	    Using \eqref{eq19} and \eqref{eq20}, \eqref{eq18} can be written as 
	         \begin{align*}
	        H\big(&\{\mathbf{F}_{(n,t_i)}, n\in\mathcal{H}_{\gamma_i}\}\big)
	         + H\big(\{\mathbf{S}_{(\gamma'_i,t_i)},\gamma'_i\in\mathcal{F}_{(\gamma_i,\text{child})}\}\big)\\
	         -H\big(&\{{\mathbf{Z}'}_n^{(t_i)}, n\in\mathcal{H}_{\gamma_i}\}\big)
	        -H\big(\{\tilde{\mathbf{Z}}_{(\gamma'_i,t_i)},\gamma'_i\in\mathcal{F}_{(\gamma_i,\text{child})}\}\big)\le 0,
	        \end{align*}
	        where the last term follows from the fact that uniform variables maximize entropy. 
	        
	        Thus, $	I\big(\mathcal{W}_{N\backslash\mathcal{T}};\mathcal{M}_{\tilde{U}_i}\big|\mathcal{K}_{N,\mathcal{T}},\{\mathcal{M}_{\tilde{U}_j},j\in\mathcal{J}_i\} \big)=0$, where
	        $\mathcal{J}_i=\{j: \gamma_j\in\mathcal{F}_{(\gamma_i,\text{desc})}, \text{ for } j\in[T]\}$. 
	\end{proof}
	\end{lemma}
	According to Lemma~\ref{lemma7}, \eqref{25} can be written as follows.
	\begin{align}\nonumber
	    I\big(&\mathcal{W}_{N\backslash\mathcal{T}};\mathcal{M}_{\mathcal{T}},\mathcal{S}_{\Gamma}\big|\mathcal{K}_{N,\mathcal{T}} \big)= I\big(\mathcal{W}_{N\backslash\mathcal{T}}; \mathcal{S}_{\Gamma}\big|\mathcal{K}_{N,\mathcal{T}}, \mathcal{M}_{\mathcal{T}}\big)\\\nonumber
	    &+\sum_{i=1}^T
	    I\big(\mathcal{W}_{N\backslash\mathcal{T}};\mathcal{M}_{\tilde{U}_i}\big|\mathcal{K}_{N,\mathcal{T}},\{\mathcal{M}_{\tilde{U}_j},j\in\mathcal{J}_i\} \big)\\\nonumber
	    =&I\big(\mathcal{W}_{N\backslash\mathcal{T}}; \{\tilde{\mathbf{Z}}_{(\Gamma,t')},t'\in[\nu]\backslash\{\mathcal{T}\cup\mathcal{D}\}\}\big|\mathcal{K}_{N,\mathcal{T}}, \mathcal{M}_{\mathcal{T}}\big)\nonumber\\ \label{28}
	    =&H\big(\{\tilde{\mathbf{Z}}_{(\Gamma,t')},t'\in[\nu]\backslash\{\mathcal{T}\cup\mathcal{D}\}\}\big|\mathcal{K}_{N,\mathcal{T}}, \mathcal{M}_{\mathcal{T}}\big)\\\nonumber
	    &-H\big(\{\tilde{\mathbf{Z}}_{(\Gamma,t')},t'\in[\nu]\backslash\{\mathcal{T}\cup\mathcal{D}\}\}\big|\mathcal{K}_{N,\mathcal{T}}, \mathcal{R}_{\mathcal{T}},\mathcal{W}_{N}\big)=0,
	\end{align}
	where  $\mathcal{R}_{\mathcal{T}}\triangleq\bigcup_{i\in[T]}\mathcal{R}_{\tilde{U}_i}$, and 
  $\mathcal{R}_{\tilde{U}_i}~\triangleq~\big\{\{{\mathbf{Z}'}_n^{(t_i)},n\in\mathcal{H}_{\gamma_i}\},\{\tilde{\mathbf{Z}}_{(\gamma'_i,t_i)},\gamma'_i\in\mathcal{F}_{(\gamma_i,\text{child})}\}\big\}$. 
Using argument similar to that in the proof of Lemma~\ref{lemma7}, independence of local models and random vectors, and according to Lemma~\ref{lemma1}, Lemma~\ref{lemma4} and Lemma~\ref{lemma6}, both terms in \eqref{28} are equal to $H\big(\{\tilde{\mathbf{Z}}_{(\Gamma,t')},t'\in[\nu]\backslash\{\mathcal{T}\cup\mathcal{D}\}\}\big)$ and the result is 0. Therefore, the privacy constraint is satisfied, i.e.,
\begin{align*}
    I\big(\mathcal{W}_{N\backslash\mathcal{T}};\mathcal{M}_{\mathcal{T}},\mathcal{S}_{\Gamma},\{\mathbf{W}_k,\mathcal{Z}_k,{k\in\mathcal{T}}\}\big|\sum_{n\in[N]\backslash\{\mathcal{D}\cup\mathcal{T}\}}{\mathbf{W}_n} \big)=0.
\end{align*}
	\section{conclusion}
In this paper, we propose \swift, a secure aggregation protocol for model aggregation in federated learning, which acheives a trade-off between the communication loads and network connections. Via partitioning the users into groups,  careful designs of intra and inter group secret sharing and aggregation method, \swift is able to achieve correct aggregation in presence of up to $D$ dropout users, with the worst-case security guarantee against up to $T$ users colluding with a curious server.
Moreover, \swift has a flexibility to control the delay of each training iteration by choosing the aggregation tree with different depth. 
Compared with previous secure aggregation protocols, \swift significantly slashes the communication load. 
 For $T=o(N)$ and $N-D=\mathcal{O}(N)$, in terms of the number of required bits, the server communication load in \texttt{SwiftAgg+} is within factor 1 of the cut-set lower bound.
In addition,    in the case of local models with uniform distributions, the per-user communication load in  \texttt{SwiftAgg+} is within factor $\log_{\ell}{\ell N} $ of the cut-set lower bound as long as $T=o(N)$ and $D=o(N)$.

	\bibliographystyle{ieeetr}
	\bibliography{References}

\begin{thebibliography}{10}

\bibitem{mcmahan2017communication}
B.~McMahan, E.~Moore, D.~Ramage, S.~Hampson, and B.~A. y~Arcas,
  ``Communication-efficient learning of deep networks from decentralized
  data,'' in {\em Artificial intelligence and statistics}, pp.~1273--1282,
  PMLR, 2017.

\bibitem{kairouz2019advances}
P.~Kairouz, H.~B. McMahan, B.~Avent, A.~Bellet, M.~Bennis, A.~N. Bhagoji,
  K.~Bonawitz, Z.~Charles, G.~Cormode, R.~Cummings, {\em et~al.}, ``Advances
  and open problems in federated learning,'' {\em arXiv preprint
  arXiv:1912.04977}, 2019.

\bibitem{li2020federated}
T.~Li, A.~K. Sahu, A.~Talwalkar, and V.~Smith, ``Federated learning:
  Challenges, methods, and future directions,'' {\em IEEE Signal Processing
  Magazine}, vol.~37, no.~3, pp.~50--60, 2020.

\bibitem{zhu2020deep}
L.~Zhu and S.~Han, ``Deep leakage from gradients,'' in {\em Federated
  learning}, pp.~17--31, Springer, 2020.

\bibitem{geiping2020inverting}
J.~Geiping, H.~Bauermeister, H.~Dr{\"o}ge, and M.~Moeller, ``Inverting
  gradients--how easy is it to break privacy in federated learning?,'' {\em
  arXiv preprint arXiv:2003.14053}, 2020.

\bibitem{bonawitz2017practical}
K.~Bonawitz, V.~Ivanov, B.~Kreuter, A.~Marcedone, H.~B. McMahan, S.~Patel,
  D.~Ramage, A.~Segal, and K.~Seth, ``Practical secure aggregation for
  privacy-preserving machine learning,'' in {\em proceedings of the 2017 ACM
  SIGSAC Conference on Computer and Communications Security}, pp.~1175--1191,
  2017.

\bibitem{so2021turbo}
J.~So, B.~G{\"u}ler, and A.~S. Avestimehr, ``Turbo-aggregate: Breaking the
  quadratic aggregation barrier in secure federated learning,'' {\em IEEE
  Journal on Selected Areas in Information Theory}, vol.~2, no.~1,
  pp.~479--489, 2021.

\bibitem{bell2020secure}
J.~H. Bell, K.~A. Bonawitz, A.~Gasc{\'o}n, T.~Lepoint, and M.~Raykova, ``Secure
  single-server aggregation with (poly) logarithmic overhead,'' in {\em
  Proceedings of the 2020 ACM SIGSAC Conference on Computer and Communications
  Security}, pp.~1253--1269, 2020.

\bibitem{choi2020communication}
B.~Choi, J.-y. Sohn, D.-J. Han, and J.~Moon, ``Communication-computation
  efficient secure aggregation for federated learning,'' {\em arXiv preprint
  arXiv:2012.05433}, 2020.

\bibitem{yu2019lagrange}
Q.~Yu, S.~Li, N.~Raviv, S.~M.~M. Kalan, M.~Soltanolkotabi, and S.~A.
  Avestimehr, ``Lagrange coded computing: Optimal design for resiliency,
  security, and privacy,'' in {\em The 22nd International Conference on
  Artificial Intelligence and Statistics}, pp.~1215--1225, PMLR, 2019.

\bibitem{kadhe2020fastsecagg}
S.~Kadhe, N.~Rajaraman, O.~O. Koyluoglu, and K.~Ramchandran, ``Fastsecagg:
  Scalable secure aggregation for privacy-preserving federated learning,'' {\em
  arXiv preprint arXiv:2009.11248}, 2020.

\bibitem{yang2021lightsecagg}
J.~So, C.~J. Nolet, C.-S. Yang, S.~Li, Q.~Yu, R.~E~Ali, B.~Guler, and
  S.~Avestimehr, ``Lightsecagg: a lightweight and versatile design for secure
  aggregation in federated learning,'' {\em Proceedings of Machine Learning and
  Systems}, vol.~4, pp.~694--720, 2022.

\bibitem{zhao2021information}
Y.~Zhao and H.~Sun, ``Information theoretic secure aggregation with user
  dropouts,'' in {\em 2021 IEEE International Symposium on Information Theory
  (ISIT)}, pp.~1124--1129, IEEE, 2021.

\bibitem{jahani2022swiftagg}
T.~Jahani-Nezhad, M.~A. Maddah-Ali, S.~Li, and G.~Caire, ``Swiftagg:
  Communication-efficient and dropout-resistant secure aggregation for
  federated learning with worst-case security guarantees,'' in {\em 2022 IEEE
  International Symposium on Information Theory (ISIT)}, pp.~103--108, 2022.

\bibitem{shamir1979share}
A.~Shamir, ``How to share a secret,'' {\em Communications of the ACM}, vol.~22,
  no.~11, pp.~612--613, 1979.

\bibitem{schlegel2021codedpaddedfl}
R.~Schlegel, S.~Kumar, E.~Rosnes, {\em et~al.}, ``Codedpaddedfl and
  codedsecagg: Straggler mitigation and secure aggregation in federated
  learning,'' {\em arXiv preprint arXiv:2112.08909}, 2021.

\bibitem{blakley1984security}
G.~R. Blakley and C.~Meadows, {\em Security of ramp schemes}, pp.~242--268.
\newblock 1984.

\bibitem{kedlaya2011fast}
K.~S. Kedlaya and C.~Umans, ``Fast polynomial factorization and modular
  composition,'' {\em SIAM Journal on Computing}, vol.~40, no.~6,
  pp.~1767--1802, 2011.

\bibitem{bonawitz2016practical}
K.~Bonawitz, V.~Ivanov, B.~Kreuter, A.~Marcedone, H.~B. McMahan, S.~Patel,
  D.~Ramage, A.~Segal, and K.~Seth, ``Practical secure aggregation for
  federated learning on user-held data,'' {\em arXiv preprint
  arXiv:1611.04482}, 2016.

\bibitem{aigner1999proofs}
M.~Aigner and G.~M. Ziegler, ``Proofs from the book,'' {\em Berlin. Germany},
  vol.~1, 1999.

\end{thebibliography}


 




\vfill

\end{document}